%% file: main.tex
\documentclass[sigconf]{acmart}

\copyrightyear{2021} 
\acmYear{2021} 
\setcopyright{acmlicensed}\acmConference[SIGMOD '21]{Proceedings of the 2021 International Conference on Management of Data}{June 20--25, 2021}{Virtual Event, China}
\acmBooktitle{Proceedings of the 2021 International Conference on Management of Data (SIGMOD '21), June 20--25, 2021, Virtual Event, China}
\acmPrice{15.00}
\acmDOI{10.1145/3448016.3457298}
\acmISBN{978-1-4503-8343-1/21/06}

\ccsdesc[500]{Theory of computation~Data structures and algorithms for data management}

\keywords{Personalized PageRank; Power Iteration; Forward Push}

\usepackage{booktabs} 

\usepackage{amsmath}
\usepackage{balance}
\usepackage{euscript}
\usepackage{graphicx}
\usepackage{subcaption}
\usepackage{hyperref}
\usepackage{mathrsfs}
\usepackage{mdframed}
\usepackage{microtype} 
\usepackage{multirow}
\usepackage[skins,breakable]{tcolorbox}
\usepackage[normalem]{ulem}
\usepackage{xspace}
\usepackage{enumitem}
\usepackage[ruled, vlined, linesnumbered]{algorithm2e}

\newcommand{\nosemic}{\renewcommand{\@endalgocfline}{\relax}}
\newcommand{\dosemic}{\renewcommand{\@endalgocfline}{\algocf@endline}}
\let\oldnl\nl
\newcommand{\nonl}{\renewcommand{\nl}{\let\nl\oldnl}}

\input{./def/jh-def}

\def\figcapup{\vspace{-3mm}}

\newtheorem{observation}{Observation}

\def\hpi{\hat{\pi}}
\def\vpi{\vec{\pi}}
\def\hpis{\hat{\pi_s}}
\def\rmax{r_{\max}}
\def\rsum{r_{\text{sum}}}
\def\rbefore{\rsum^{\text{before}}}
\def\rafter{\rsum^{\text{after}}}
\def\pis{\vec{\pi_s}}
\def\epis{\hat{\pi_s}}
\def\abserror{\|\epis - \pis\|_1}
\def\mA{\mathbf{A}}
\def\mM{\mathbf{M}}

\def\mP{\mathbf{P}}
\def\ves{\vec{e_s}}
\def\vrs{\vec{r_s}}
\def\ve{\vec{e}}
\def\vga{\vec{\gamma_s}}
\def\itj{^{(j)}}
\def\itjnext{^{(j + 1)}}

\def\vpisjnext{\vpi_s^{(j+1)}}

\def\powitr{\emph{PowItr}}
\def\fwdpush{\emph{FwdPush}}
\def\powforpush{\emph{PowerPush}}
\def\simfwdpush{\emph{SimFwdPush}}
\def\itrfwdpush{\emph{FIFO-FwdPush}}
\def\speedppr{\emph{SpeedPPR}}

\def\fora{\emph{FORA}}
\def\mc{\emph{MonteCarlo}}
\def\bepi{\emph{BePI}}
\def\resacc{\emph{ResAcc}}
\def\speedppri{\emph{SpeedPPR-Index}}
\def\forai{\emph{FORA-Index}}

\begin{document}
\fancyhead{}

\begin{sloppy}

\title[Unifying the Global and Local Approaches: An Efficient Power Iteration with Forward Push]
{Unifying the Global and Local Approaches: \texorpdfstring{\\}{Lg} An Efficient Power Iteration with Forward Push}


\author{Hao Wu}
\email{whw4@student.unimelb.edu.au}
\affiliation{
  \institution{The University of Melbourne$^*$}
  \authornote{~School of Computing and Information Systems}
  \country{}
}

\author{Junhao Gan}
\email{junhao.gan@unimelb.edu.au}
\affiliation{
 \institution{The University of Melbourne$^*$}
  \country{}
}

\author{Zhewei Wei}
\email{zhewei@ruc.edu.cn}
\affiliation{
	\institution{Renmin University of China$^\dagger$}
    \authornote{Gaoling School of Artificial Intelligence}
     \country{}
}

\author{Rui Zhang}
\email{rui.zhang@ieee.org}
\affiliation{
 \institution{Tsinghua University$^\ddagger$}
 \authornote{\href{https://ruizhang.info}{www.ruizhang.info}}
  \country{}
}

\begin{abstract}
	{\em Personalized PageRank} (PPR) is a critical measure of the {\em importance} of a node $t$ to a source node $s$ in a graph.
	The {\em Single-Source PPR} (SSPPR) query computes the PPR's of all the nodes with respect to $s$ on a directed graph $G$ with $n$ nodes and $m$ edges; 
	and it is an essential operation widely used in graph applications.
	In this paper, we propose novel algorithms for answering two variants of SSPPR queries: (i) {\em high-precision} queries and (ii) {\em approximate} queries.
	
	For high-precision queries, {\em Power Iteration} ($\powitr$) and {\em Forward Push} ($\fwdpush$) are two fundamental approaches.
	Given an absolute error threshold $\lmd$ (which is typically set to as small as $10^{-8}$), the only known bound of $\fwdpush$ is $O(\frac{m}{\lmd})$, much worse than the $O(m \log \frac{1}{\lmd})$-bound of $\powitr$.
	Whether $\fwdpush$ can achieve the same running time bound as $\powitr$ does still remains an open question in the research community.
	We give a positive answer to this question. We show that 
	the running time of a common implementation of $\fwdpush$ is actually bounded by $O(m \cdot \log \frac{1}{\lmd})$.
	Based on this finding, we propose a new algorithm, called {\em Power Iteration with Forward Push} ($\powforpush$), which incorporates the strengths of both $\powitr$ and $\fwdpush$.
	
	For approximate queries (with a relative error $\eps$), 
	we propose a new algorithm, called $\speedppr$,  with overall expected time bounded by $O(n \cdot \log n \cdot \log \frac{1}{\eps})$ on scale-free graphs.
	This improves the state-of-the-art $O(\frac{n \cdot \log n}{\eps})$ bound.

	We conduct extensive experiments on six real datasets.
	The experimental results show that $\powforpush$ outperforms the state-of-the-art high-precision algorithm $\bepi$ by up to an order of magnitude in both efficiency and accuracy. 
	Furthermore, our $\speedppr$ also outperforms the state-of-the-art approximate algorithm $\fora$ by up to an order of magnitude in all aspects including query time, accuracy, pre-processing time as well as index size. 

\end{abstract}

\maketitle

\input{introduction.tex}

\input{problem.tex}

\input{pre-high.tex}

\input{analysis.tex}

\input{powforpush.tex}

\input{speedppr.tex}

\input{related-work.tex}





\input{experiment}

\input{conclusion}



\vspace{-2mm}
\section*{Acknowledgement}
In this research, Junhao Gan was in part supported by Australian Research Council (ARC)
Discovery Early Career Researcher Award (DECRA) DE190101118.
Zhewei Wei  was supported in part by National Natural Science Foundation of China (No. 61972401, No. 61932001 and No. 61832017), and by Beijing Outstanding Young Scientist Program NO. BJJWZYJH012019100020098. Zhewei Wei also works at Beijing Key Laboratory of Big Data Management and Analysis Methods, MOE Key Lab of Data Engineering and Knowledge Engineering, and Pazhou Lab, Guangzhou, 510330, China.
Both Junhao Gan and Zhewei Wei are the corresponding authors.

\bibliographystyle{ACM-Reference-Format}
\balance
\bibliography{acmart.bib}  

\end{sloppy}
\end{document}

%% file: def/jh-def.tex
\usepackage{amsfonts}

\def\figcapup{\vspace{-1mm}}

\def\eps{\epsilon}
\def\lmd{\lambda}

\def\-{\mbox{-}}

\def\la{\langle}

\def\ra{\rangle}

\def\*{\star}



%% file: introduction.tex
\section{Introduction}\label{sec:intro}

As a natural data model, graphs are playing a more and more important role in real-world applications nowadays.
In a graph, it is often useful to measure the {\em relevance} between nodes.
One of the most important relevance measurements is the {\em importance} of a node $t$ to a node $s$,
for which the {\em Personalised PageRank} (PPR) is a widely adopted indicator.

Consider a directed graph $G = \la V, E \ra$ with $n$ nodes and $m$ edges, a source node $s$ and a target node $t$ in $V$;
the \textit{PPR} of $t$ with respect to $s$, denoted by $\pi(s, t)$, 
is the probability that an $\alpha$-random walk from $s$ stops at $t$. 
Specifically, an $\alpha$-random walk (for some constant $\alpha \in [0, 1)$, e.g., $\alpha = 0.2$) from $s$ is proceeded as follows: 
starting from $s$, the walk may stop at the current node $v$ (initially $v=s$) with the probability of $\alpha$, or with the probability of $1 - \alpha$, the walk may move to one of $v$'s out-neighbors uniformly at random.

Of particular interest is the {\em Single Source PPR} ({SSPPR}) query;  its goal is to compute $\pi(s,v)$ for every node $v \in V$ with respect to a given source node $s$. 
The answer to a SSPPR query is a {\em vector} in $\mathbb{R}^{1 \times n}$, denoted by $\pis$, of which the $i$-th coordinate is the PPR $\pi(s, v_i)$, where $v_i$ is the $i$-th node in $G$.
The SSPPR query has many important traditional applications such as computing PageRank and Who-to-Follow recommendation in social networks (e.g., Twitter). Moreover, the SSPPR query provides essential and primitive features widely used in representation learning for graphs, which is attracting huge attention in the machine learning community at the moment.
For example, the PPR information has been adopted in graph embedding methods such as HOPE~\cite{OuCPZ016}, STRAP~\cite{YinW19} and Verse~\cite{TsitsulinMKM18}, and 
graph attention networks such as ADSF~\cite{0001ZWZ20}.

Therefore, it is imperative to have highly efficient algorithms for answering SSPPR queries. 
It is known that an SSPPR query can be {\em precisely} solved by solving the following linear equation system~\cite{page1999pagerank}:
{\small
\begin{align}\label{eq:pprsystem}
	\pis = \alpha \cdot \ves + (1 - \alpha) \cdot \pis \cdot \mP\,,
\end{align}
}%
where $\ves \in \mathbb{R}^{1 \times n}$ is an indicator vector which has $1$ on the $s$-th dimension and $0$ for others,
and $\mP \in \mathbb{R}^{n\times n}$ is the so-called {\em transition matrix} of $G$.
However, solving Equation~\eqref{eq:pprsystem} requires to compute the inverse of an $n\times n$  matrix related to $\mP$, which is expensive. 
In practice, to trade for better efficiency, people instead compute an estimation $\hpis$ of $\pis$, which meets a certain error criteria.
Along this direction, SSPPR queries can be categorized into two variants: (i) {\em High-Precision SSPPR} queries and (ii) {\em Approximate SSPPR} queries.

In this paper, we propose novel algorithms for answering these two types of queries.
Our algorithms are efficient both in theory and in practice.
Before illustrating our results, we first set up the context of the relevant state-of-the-art algorithms.

\vspace{1mm}
\noindent
{\bf High-Precision SSPPR.} 
The goal of this type of queries is to compute a {\em high-precision} estimation $\epis$ of $\pis$ such that the $\ell_1$-error $\abserror \leq \lmd$, 
where $\lmd$ is a specified threshold and $\lmd$ is often set to as small as $10^{-8}$. 
{\em Power Iteration} ($\powitr$) and {\em Forward Push} ($\fwdpush$) are two fundamental approaches to answer high-precision SSPPR queries. 


\vspace{1mm}
\noindent
\underline{\em Power Iteration ($\powitr$).}
$\powitr$ is an iterative algorithm for solving Equation~\eqref{eq:pprsystem}.
More specifically, it refines an estimation $\hpis$ of $\pis$ iteration by iteration;
in each iteration, $\abserror$ decreases by a factor of $(1- \alpha)$.
It is known that the  overall running time of $\powitr$ is bounded by $O(m \cdot \log \frac{1}{\lmd})$ \cite{Berkhin05}.

\vspace{1mm}
\noindent
\underline{\em Forward Push ($\fwdpush$).}
$\fwdpush$ is another feasible approach to answer high-precision SSPPR queries.
It is well-known that the running time of $\fwdpush$ is bounded by $O(\frac{1}{\rmax})$, where $\rmax$ is a parameter that controls the stop condition of the algorithm.
However, 
at the time when $\fwdpush$ was first proposed in 2006~\cite{AndersenCL06}, the $\ell_1$-error bound of this approach was unclear. 
In 2017, Wang et. al \cite{WYXWY17} officially documented that the $\ell_1$-error is bounded by
{\small
\begin{align}\label{eq:mrmax}
\|\hat{\pi_s} - \vec{\pi_s}\|_1 \leq m \cdot r_{\max} \,.
\end{align}
}%
Therefore, in order to guarantee $\abserror \leq \lmd$, one needs to set $\rmax = \lmd / m$ leading to an overall time complexity $O(m / \lmd)$.
Unfortunately, this bound is not quite useful. Given that $\lmd$ is often as small as $10^{-8}$, this bound would imply a huge cost when the graph is large, e.g., on the billion-edge Twitter graph.
However, interestingly, despite of the $O(m /\lmd)$-bound, $\fwdpush$ is found to
be more efficient
than the bound suggests in certain applications (e.g., the Approximate SSPPR queries as discussed below). 

Therefore, a significant {\em knowledge gap} still exists between the practical use and the theoretical understanding of $\fwdpush$.
In particular, the following question:
\begin{quote}
	{\em Does $\fwdpush$ admit a tighter running time bound with a weaker dependency on the $\ell_1$-error threshold $\lmd$?}
\end{quote}
remains {\em open} to the research community.

\vspace{1mm}
\noindent
{\bf Approximate SSPPR.}
The aim of approximate SSPPR is to compute an estimation $\hpi(s, v)$ bounded by a {\em relative error} $\eps$, i.e., $|\hpi(s,v) - \pi(s, v)| \leq \eps \cdot \pi(s,v)$,
for every node $v$ whose $\pi(s,v) \geq 1/n$, and the algorithm must be correct with probability of at least $1 - {1}/{n}$.

\vspace{1mm}
\noindent
\underline{\em FORA.} 
$\fora$ \cite{WangTXYL16} is a representative of the state-of-the-art approximate SSPPR algorithms.
%
The basic idea of $\fora$ is to combine $\fwdpush$ and the $\mc$  method.
Specifically, there are two phases:
in the first phase, 
$\fora$ runs $\fwdpush$ to obtain an estimation $\hat{\pi_s}$ such that $ \abserror \leq m \cdot r_{\max}$.
In the second phase, the $\mc$ method based on $\hpis$ is adopted to refine the estimations to 
be within a relative error $\eps$ for every node $v$ with $\pi(s,v) \geq 1/n$. 
The overall expected running time is bounded by $O(\frac{1}{r_{\max}} + m \cdot r_{\max} \cdot \frac{n \log n}{\eps^2})$.
By setting $r_{max}$ carefully to ``balance'' the two terms and assuming the graph is {\em scale-free}, i.e., $m = O( n\log n)$,
the complexity is minimized to $O(\frac{n \log n}{\eps})$.
%
In the literature, no existing work~\cite{WangTXYL16, wei2018topppr, YWXWLY019, lin2020index} can overcome this $O(\frac{n \log n}{\eps})$-barrier.

Besides, $\fora$ admits an index version, called $\fora$+, where the results of the $\alpha$-random walks that would be needed in the $\mc$  phase are pre-generated.
With the index, the actual running time of $\fora$+ can be further reduced.
However, 
since $\fora$ has to set $\rmax$ to minimize the complexity,
the number of random walks required to be pre-generated in $\fora$+ depends on the relative error $\eps$.
Thus, the index constructed for one $\eps$ value may not be sufficient for answering a query with another smaller $\eps$ value.
This weakness significantly limits the applicability of $\fora$+.

\vspace{1mm}
\noindent
{\bf Our Contributions.} We make the following contributions:
\vspace{-1mm}
\begin{itemize}[leftmargin = *]
	\item \underline{An Equivalence Connection.} 
		We show that there essentially exists an {\em equivalence connection} between the global-approach $\powitr$ and the local-approach $\fwdpush$.
	\vspace{1mm}
	\item \underline{A Positive Answer to the Open Question.} 
		Embarking from this connection, we prove that 
		the running time of a common $\fwdpush$ implementation is actually bounded by $O(m \cdot \log \frac{1}{\lambda})$ with $\rmax  = \lmd / m$, 
		rather than the widely accepted $O(\frac{m}{\lmd})$-bound. 
			
	\vspace{1mm} 
	\item \underline{A New Algorithm for High-Precision SSPPR.} 
		Based on our finding, we 
		propose a new implementation for $\powitr$ (and hence, also for $\fwdpush$), called
		{\em Power Iteration with Forward Push} ($\powforpush$).
		Our $\powforpush$ is carefully designed such that it incorporates both the strengths of $\powitr$ and $\fwdpush$ (detailed discussions are in Section~\ref{sec:powforpush}).
		Therefore, it outperforms $\powitr$ and $\fwdpush$ in all cases in our experiments while still achieving the $O(m \cdot \log \frac{1}{\lmd})$ theoretical bound.
	
		\hspace{6pt}
		Moreover, unlike the state-of-the-art algorithm, $\bepi$~\cite{jung2017bepi},
		which requires a substantial pre-processing time and space for index storage,  
		$\powforpush$ is completely on-the-fly without needing any pre-processing or index pre-computation.
		Even though the advantage of pre-processing is taken, 
		in our experiment, on a medium-size data, {\em Orkut}, 
		$\bepi$ requires  $672$ seconds for a query. 
		Our $\powforpush$ answers the same query in less than $40$ seconds, $17$ times faster than $\bepi$.

		\hspace{6pt}
		Besides, although $\powforpush$ is a high-precision algorithm, in our experiments, in some cases, it even outperforms the state-of-the-art approximate SSPPR algorithms in running time.

		\hspace{6pt}
		Finally, given the fact that $\powitr$ is an important fundamental method, we believe that our $\powforpush$ would be of independent interests in other  applications beyond  the SSPPR queries.

	\vspace{1mm}
	\item \underline{A New Algorithm for Approximate SSPPR.}
		Based on the support of $\powforpush$, we further design a new algorithm, called {\em SpeedPPR}, for answering approximate SSPPR queries. 
		On scale-free graphs with $m = O(n\cdot\log n)$, the overall expected time of $\speedppr$ is bounded by $O(n\cdot \log n \cdot \log \frac{1}{\eps})$,
		improving the state-of-the-art $O(\frac{n \cdot \log n}{\eps})$-bound.
		Furthermore, $\speedppr$ {\em always} admits an index of at most $m$ $\alpha$-random walk results. 
				Hence, the space consumption of the index is at most as large as the graph itself. 
				More importantly, the index size of {\em SpeedPPR} is independent to the values of $\eps$. 
				In other words, once the index is built, it suffices to answer queries with any $\eps$. 
				This feature of {\em SpeedPPR} is considered as an important improvement over $\fora$+.
				In particular, for small $\eps$ values, {\em SpeedPPR} consumes $10$ times less space than $\fora$+ does for index storage.
	\vspace{1mm}
	\item \underline{Extensive Experiments.}
		We conduct extensive experiments on six real datasets which are widely adopted in the literature. 
		The experimental results show that our $\powforpush$ outperforms the state-of-the-art high-precision SSPPR algorithms by up to an order of magnitude. Our {\em SpeedPPR} outperforms all the state-of-the-art competitors for approximate SSPPR by up to an order of magnitude in terms of 
		query efficiency and result accuracy; for index-based version, $\speedppr$ also achieves up to $10$ times  improvements on both pre-processing time and index size. 

\end{itemize}

\vspace{-2mm}
\noindent
{\bf Paper Organization.}
Section~\ref{sec:problem} defines the problems and notations. 
Section~\ref{sec:pre-high} introduces $\powitr$ and $\fwdpush$ in detail.
In Section~\ref{sec:analysis}, we show the time complexity of $\fwdpush$.
In Section~\ref{sec:powforpush}, $\powforpush$ is proposed along with some crucial optimizations.
Section~\ref{sec:speedppr} shows our $\speedppr$.
Section~\ref{sec:related-work} is about related work and Section~\ref{sec:exp} shows the experimental results.
Finally, Section~\ref{sec:conclusion} concludes the paper.

%% file: problem.tex
\vspace{-3mm}
\section{Problem Formulation}\label{sec:problem}



Consider a directed graph $G = \la V, E \ra$ with $n = |V|$ nodes and $m = |E|$ edges.
Without loss of generality, we assume that the nodes in $V$ are in order such that $v_i$ is the $i$-th node in $V$, where $i \in [n]$ and $[n] = \{1, 2, \ldots, n\}$.
For a node $v \in V$, denote the set of the {\em out-neighbors} of $v$  
by $N_{out}(v) = \{u \mid (v, u) \in E\}$, and $d_v = |N_{out}(v)|$ is defined as the {\em out-degree} of $v$. 
Clearly, $m = \sum_{v\in V} d_v$.
In this paper, we assume that there is no ``dead-end'' nodes, i.e., $d_v \geq 1$ holds for all $v \in V$, in the graph $G$.
As we explain below, this assumption is without loss of generality.

\vspace{1mm}
\noindent
{\bf Indicator Vector.}
Denote by $\ve_{v_i} \in \mathbb{R}^{1 \times n}$ the {\em indicator vector} which has coordinate $1$ on the $i$-th dimension and $0$ on the others, where $v_i \in V$.
It is easy to verify that for any $n\times n$ matrix $\mM$, the result of $\ve_{v_i} \cdot \mM$ is exactly the $i$-th row of $\mM$.

\vspace{1mm}
\noindent
{\bf $\ell_1$-Norm.}
For any $n$-dimensional vector $\vec{x}$, the $\ell_1$-norm of $\vec{x}$ is computed as $\|\vec{x}\|_1 = \sum_{i = 1}^n |x_i|$, where $x_i$ is the $i$-th coordinate of $\vec{x}$.

\vspace{1mm}
\noindent
{\bf Adjacent Matrix.}
The adjacent matrix $\mA$ of a directed graph $G$ is an $n\times n$ matrix, 
where the $i$-th row of $\mA$, denoted by $\vec{A}_{v_i}$, is a row vector which has coordinate $1$ on the $j$-th dimension if $(v_i, v_j) \in E$ and $0$ otherwise, for $j \in [n]$.

\vspace{1mm}
\noindent
{\bf Transition Matrix.} 
The transition matrix $\mP$ of a directed graph $G$ with an adjacent matrix $\mA$ is 
an $n\times n$ matrix, 
where the $i$-th row of $\mP$, denoted by $\vec{P}_{v_i}$, satisfies $\vec{P}_{v_i} = \frac{1}{d_{v_i}} \cdot \vec{A}_{v_i}$, 
and hence, $\|\vec{P}_{v_i}\|_1 = 1$  for all $i \in [n]$.
Furthermore, it can be verified that for any vector $\vec{x} \in \mathbb{R}^{1 \times n}$, it holds that $\|\vec{x} \cdot \mP^k\|_1 = \|\vec{x}\|_1$ for all integer $k \geq 1$. An example of a transition matrix is shown in Figure~\ref{fig:graph}.

\begin{figure}[t]
	\begin{minipage}{0.45\linewidth}
		\centering
		\includegraphics[height = 22mm]{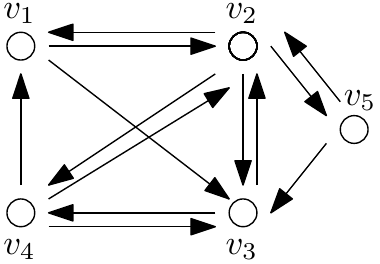}
	\end{minipage}%
	\begin{minipage}{0.5\linewidth}
	\begin{align*}
	\mP = 
	\begin{bmatrix}
	0   & 1/2 & 1/2 & 0   & 0 \\
	1/4 &  0  & 1/4 & 1/4 & 1/4 \\
	0   & 1/2 & 0   & 1/2 & 0 \\
	1/3 & 1/3 & 1/3 & 0   & 0 \\
	0   & 1/2 & 1/2 & 0   & 0
	\end{bmatrix}
	\end{align*}
	\end{minipage}
	\vspace{-2mm}
	\caption{A directed graph $G$ and its transition matrix $\mP$}
	\label{fig:graph}
	\vspace{-5mm}
\end{figure}

\vspace{1mm}
\noindent
{\bf $\alpha$-Random Walk.} 
Consider a {\em constant} parameter $\alpha \in [0, 1)$ which is set to $0.2$ by default in the literature; 
an $\alpha$-random walk from a node $s \in V$ is defined as follows:
let $v$ be the current node and initially the current node $v$ is $s$; 
at every step, the walk stops at $v$ with probability $\alpha$, 
and with probability $1 - \alpha$, the walk moves one-step forward depending on either of the following two cases:
(i) if $N_{out} \neq \emptyset$, the walk {\em uniformly at random}, i.e., with equal probability $\frac{1}{d_v}$, moves to an out-neighbor $u \in N_{out}(v)$ 
(that is, the current node $v$ now becomes $u$);
(ii) otherwise (i.e., $N_{out} = \emptyset$), the walk jumps back to $s$ (the current node $v$ becomes $s$).
Effectively, this is equivalent to {\em conceptually} add an edge from each ``dead-end'' node (whose out-degree is $0$) to the source node $s$,
and hence, one can assume that no dead-end node exists in the graph.
Moreover, without stated otherwise, all the {\em random walks} considered in this paper are $\alpha$-random walks.

\vspace{1mm}
\noindent
{\bf Alive Random Walk.}
If an $\alpha$-random walk at the current node $v$ does not stop yet, then we say this random walk is {\em alive} at $v$.

\vspace{1mm}
\noindent
{\bf Personalized PageRank (PPR).} Consider a node $s \in V$ and a node $t \in V$; the PPR of $t$ with respect to $s$, denoted by $\pi(s, t)$,
is defined as the probability that an $\alpha$-random walk from $s$ stops at $t$.

\vspace{1mm}
\noindent
{\bf Single Source Personalized PageRank (SSPPR).}
Given a source node $s \in V$, the goal of a SSPPR query is to compute 
the {\em PPR} vector $\pis$, where the $i$-th coordinate in $\pis$ is the PPR $\pi(s, v_i)$ of $v_i$.
Essentially, $\pis$ is the {\em probability distribution} over all the nodes that an $\alpha$-random walk from $s$ stops at a node.
Thus, $\|\pis\|_1 = 1$.

\vspace{1mm}
\noindent
{\bf High-Precision SSPPR (HP-SSPPR).}
Given an $\ell_1$-error threshold $\lmd \in (0, 1]$, the goal of a High-Precision SSPPR query is to compute an estimation $\hpis$ of $\pis$ such that
$\abserror \leq \lmd$.
In general, the value of $\lmd$ is set to $\min\{\frac{1}{m}, 10^{-8}\}$.

\vspace{1mm}
\noindent
{\bf Approximate SSPPR (Approx-SSPPR).}
Given an relative error threshold $\eps > 0$ and a PPR value threshold $\mu \in (0, 1]$, an Approximate SSPPR query aims to compute
an estimation $\hpi(s, v)$ for each node $v$ with $\pi(s, v) \geq \mu$ such that $| \hpi(s, v) - \pi(s, v)| \leq \eps \cdot \pi(s,v)$
with high probability $1 - \frac{1}{n}$.
In the literature, $\mu$ is conventionally set to the average over all the PPR values with respect to $s$, i.e., $\frac{1}{n}$.

%% file: pre-high.tex
\begin{figure*}
	\centering
	\includegraphics[width = 0.7\linewidth]{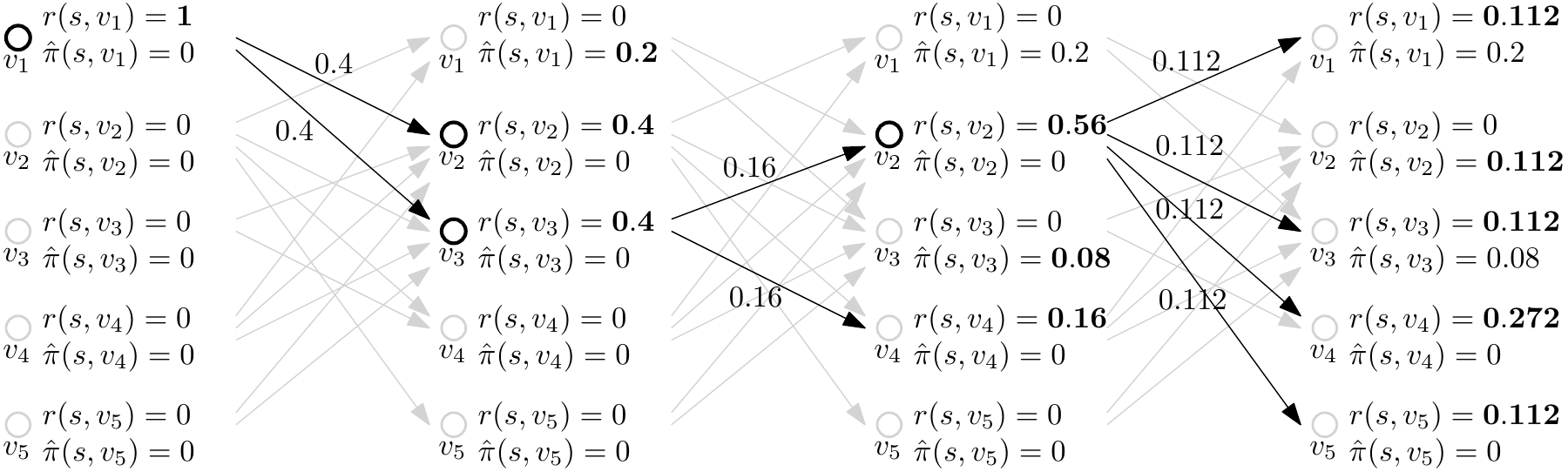}
	\vspace{-2mm}
	\caption{A running example of Forward Push on $G$ in Figure~\ref{fig:graph} with $s = v_1$, $\alpha =0.2$ and $\rmax = 0.099$. 
		The arrows are the edges of $G$ but presented in a bipartite sense.
		Active nodes in each iteration are highlighted in bold.
	}
	\label{fig:fwdpush}
	\vspace{-2mm}
\end{figure*}

\section{Preliminaries of HP-SSPPR}\label{sec:pre-high}

In this section, we introduce the details of the two most relevant existing approaches to this paper for answering Hight-Precision SSPPR queries: 
{\em Power Iteration} ($\powitr$) and {\em Forward Push} ($\fwdpush$). 
In the literature, 
they are often considered as two different types of methods, 
respectively as {\em global} and {\em local} approaches.
Thus, in the previous work, these two approaches are usually understood and explained from different perspectives.
%
%
Next, we explain both $\powitr$ and $\fwdpush$ from a {\em unified} perspective of alive random walks.
Our explanations would hint some ideas on the {\em equivalence connection} between these two approaches as discussed in Section~\ref{sec:analysis}.

\subsection{The Power Iteration Approach}\label{sec:powitr}


Define $\vga\itj \in \mathbb{R}^{1 \times n}$ as the vector, of which the $i$-th coordinate is the probability that 
an $\alpha$-random walk from $s$ is {\em alive} at $v_i$ at the $j$-th step.
Clearly, $\vga^{(0)} = \ves$: a random walk from $s$ can be {alive} only at $s$ at the initial state, i.e., at the $0$-th step. 
According to the definition of $\alpha$-random walks, it can be verified that 
\begin{equation}\label{eq:aliveitr} \small
	\vga\itjnext = (1 - \alpha) \cdot \vga\itj \cdot \mP
	= (1 - \alpha)^{j+1} \cdot \ves \cdot \mP^{j+1}\,.
\end{equation}
Therefore, the $i$-th coordinate of the vector
$\alpha \cdot \vga\itj$ 
is the probability that 
a random walk from $s$ stops at $v_i$ at {\em exactly} the $j$-th step.

Intuitively, by the definition of PPR, $\pi(s,t)$ can be computed as the sum of the probabilities that 
an $\alpha$-random walk from $s$ stops at $t$ with {\em exactly} $j$ steps, for all possible length $j = 0, 1, 2, \cdots$.
That is, 
\begin{equation}\label{eq:pis} \small
	\pis = \sum_{j = 0}^\infty \alpha \cdot \vga\itj\,.
\end{equation}
The basic idea of $\powitr$ is to iteratively maintain an underestimate $\hpis = \vpisjnext$ in the $(j+1)$-th iteration (for $j = 0, 1, 2, \cdots$) such that:
\begin{equation}\label{eq:pij} \small
	\vpisjnext = \sum_{k=0}^{j} \alpha \cdot \vga^{(k)} \,,
\end{equation}
where the $i$-th coordinate of $\vpisjnext$ is the probability that a random walk from $s$ stops at $v_i$ with {\em at most} $j$ steps.

\vspace{1mm}
\noindent
{\bf The $\ell_1$-Error Bound.}
The $\ell_1$-error at the end of the $(j+1)$-th iteration is given by:
\begin{align}\label{eq:powitr-error} \small
	\abserror &= \|\sum_{k = j + 1}^\infty \alpha \cdot \vga^{(k)}\|_1
		  = \sum_{k = j + 1}^\infty \alpha \cdot (1 - \alpha)^{k} \cdot \|\vec{e}_s \cdot \mP^{k}\|_1 \nonumber \\
		  &= \sum_{k = j + 1}^\infty \alpha \cdot (1 - \alpha)^{k} \cdot \|\vec{e}_s\|_1  = (1 - \alpha)^{j+1} \,.
\end{align}
Two observations follow immediately from Equation~\eqref{eq:powitr-error}. 
First, the $\ell_1$-error $(1-\alpha)^{j+1}$ is exactly equal to $\|\vga\itjnext\|_1$, 
that is, the total probability mass of a random walk alive at the $(j+1)$-th steps.
This is intuitive because this is exactly the total amount of the probability mass 
that is not yet converted to the PPR values.
Second, the $\ell_1$-error decreases by a 
factor of $(1 - \alpha)$ in each iteration; after at most $O(\log \frac{1}{\lmd})$ iterations, $\abserror \leq \lmd$.
Since it is easy to verify that the computational cost of each iteration is bounded by $O(m)$,
the overall running time of $\powitr$ is thus bounded by $O(m \cdot \log \frac{1}{\lmd})$.

\subsection{The Forward Push Approach}\label{sec:fwdpush}

$\fwdpush$ {\em conceptually} considers a random walk from $s$ and observes the state of this walk in terms of probability mass.
Given a specified parameter $\rmax \in [0, 1]$, 
the basic idea of $\fwdpush$ is to maintain, for each node $v \in V$, the following information:
\begin{itemize}[leftmargin = *]
	\item a {\em reserve} $\hpi(s,v)$: it is an underestimate of $\pi(s,v)$ and
	\item a {\em residue} $r(s,v)$: 
		it is the {\em unprocessed} probability mass of the random walk from $s$ {\em alive} at $v$ at the current state.
\end{itemize}
Initially, $\hpi(s,v) = 0$ for all $v \in V$ and $r(s,v) = 0$ for all $v \neq s$ while $r(s,s) = 1$: 
at the initial state,  
the unprocessed probability mass of the random walk from $s$ alive at $s$ is $1$. 

\vspace{1mm}
\noindent
{\bf Active Nodes.}
A node $v$ is {\em active} if it satisfies $r(s,v) > d_v \cdot \rmax$; otherwise, it is {\em inactive}. 

\vspace{1mm}
\noindent
{\bf The Push Operation.}
A crucial primitive in $\fwdpush$ is the {\em push operation}, which is to 
{\em process} a node's residue.
Specifically, a push operation on a node $v$ works as follows:
\begin{itemize}[leftmargin = *]
%
	\item First, $\alpha$ portion of $v$'s residue $r(s,v)$ is converted to $\hpi(s,v)$, i.e., $\hpi(s, v) \leftarrow \hpi(s,v) +  \alpha \cdot r(s,v)$.
		This represents the fact that
		with probability $\alpha$, the alive random walk at $v$ stops at $v$.
	\item Second, the rest $(1-\alpha)$ portion of $r(s,v)$ is {\em evenly} distributed to the residues of $v$'s out-neighbors. 
		That is, the residue of each out-neighbor of $v$ is increased by $\frac{(1-\alpha)\cdot r(s,v)}{d_v}$, 
		which is the probability that, conditioned on $r(s,v)$, the random walk at $v$ moves to this out-neighbor and is alive at this out-neighbor at the current state.
	\item Third, after the residue of $v$ is processed, $r(s,v) \leftarrow 0$, indicating that currently there is no unprocessed probability mass of the random walk from $s$ currently alive at $v$.
\end{itemize}

The process of $\fwdpush$ is to repeatedly pick an {\em arbitrary} {\em active} node,
and perform a push operation on it.
The algorithm terminates until there is no active node.
Algorithm~\ref{algo:fwdpush} shows the pseudo-code.

\vspace{-3mm}
	\begin{algorithm}[h] 
	\caption{Forward Push}
	\label{algo:fwdpush}
	\KwIn{$G$, $\alpha$, $s$, $\rmax$}
	\KwOut{an estimation $\hpis$ of $\pis$}
	$\hpi(s,v) \leftarrow 0$ and $r(s,v) \leftarrow 0$ for all $v \in V$;
	$r(s, s) \leftarrow 1$\;
	\While{there exists a node $v$ with $r(s, v) > d_v \cdot \rmax$}{
		pick an {\em arbitrary} such node $v$ with $r(s, v) > d_v \cdot \rmax$\;
		$\hpi(s,v) \leftarrow \hpi(s,v) + \alpha \cdot r(s, v)$\;
		\For{each $u \in N_{out}(v)$}{
			$r(s, u) \leftarrow r(s,u) +  \frac{(1 - \alpha) \cdot r(s, v)}{d_v}$\;
		}
		$r(s,v) \leftarrow 0$\;
 
 	}
	\Return $\hpi(s,v)$ for all $v\in V$ as a vector $\hpis$\;
\end{algorithm} 

 \vspace{-4mm}
\noindent
{\bf A Running Example.} 
Figure~\ref{fig:fwdpush} shows a running example.
At the beginning,
only $v_1$ is active; thus it is picked to perform a push operation,
in which $\hpi(s, v_1) = \alpha \cdot 1 = 0.2$ and 
the residues of $v_1$'s out-neighbors $v_2$ and $v_3$ are increased by $(1-\alpha) \cdot 1 / 2 = 0.4$, respectively.
After this push operation on $v_1$, both $v_2$ and $v_3$ are now active.
The algorithm picks one of them arbitrarily; in this example, $v_3$ is picked.
After the push operation on $v_3$, $\hpi(s,v_3) = 0.2 \cdot 0.4 = 0.08$ and each of its out-neighbors, 
i.e., $v_2$ and $v_4$, has residue increased by $0.8 \cdot 0.4 / 2 = 0.16$.
Next, $v_2$ becomes the only active node; 
after the push operation on $v_2$, no node is active and thus the algorithm terminates.

\vspace{1mm}
\noindent
{\bf The $\ell_1$-Error Bound.}
When $\fwdpush$ terminates, it holds that $r(s,v) \leq d_v \cdot \rmax$ for all $v\in V$. 
By definition, the residues 
are the probability mass of the alive random walk that are not yet converted to $\hpi(s,v)$'s.
Hence,
\begin{equation} \small
\abserror  = \sum_{v \in V} r(s,v) \leq \sum_{v \in V} d_v \cdot \rmax = m \cdot \rmax \,.
\end{equation}
In order to achieve an $\ell_1$-error at most $\lmd$,
one needs to set $\rmax \leq \frac{\lmd}{m}$. 

\vspace{1mm}
\noindent
{\bf The Open Question.}
The only known time complexity of $\fwdpush$ is 
$O(\frac{1}{\rmax})$~\cite{AndersenCL06}.
Unfortunately, this bound implies that the overall running time becomes $O(\frac{m}{\lmd})$ with $\rmax = \lmd / m$, 
which is worse than the $O(m \cdot \log \frac{1}{\lmd})$-bound of $\powitr$. 
Despite of its practicality in certain applications, it still remains an open question:
Does $\fwdpush$ admit a running time bound with a weaker dependency on $\lmd$, just like what $\powitr$ does?

%% file: analysis.tex
\section{A Tighter Analysis of Forward Push} \label{sec:analysis}

\begin{figure}
	\centering
	\includegraphics[width = \linewidth]{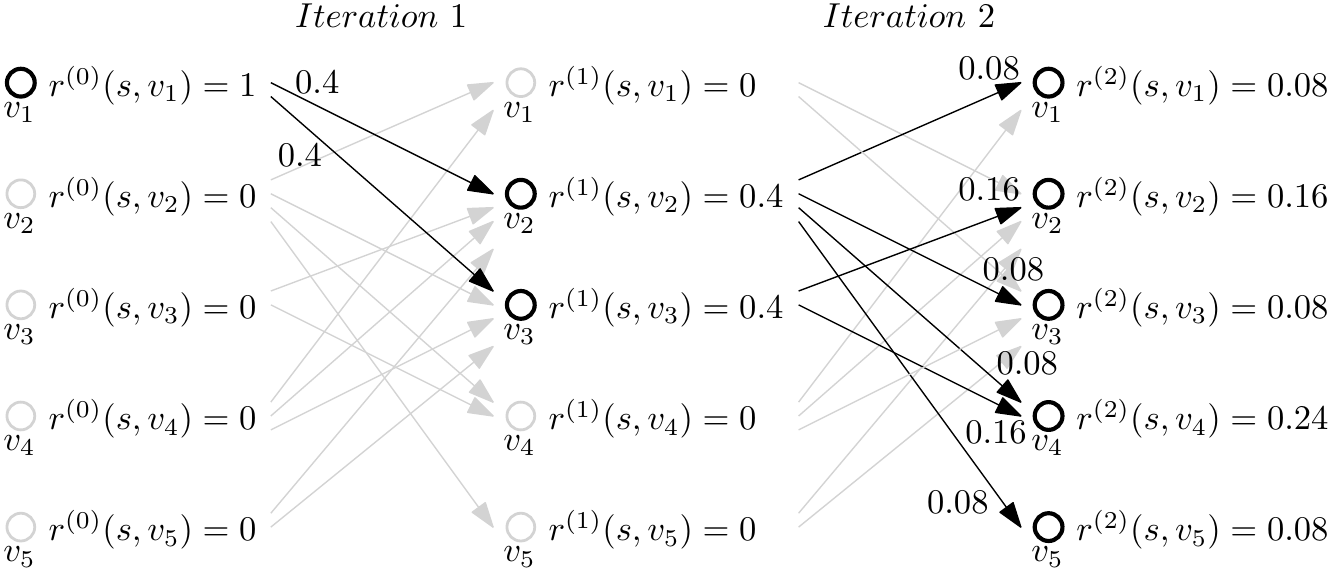}
	\vspace{-4mm}
	\caption{A running example of $\simfwdpush$ on $G$ with $s = v_1$, $\alpha =0.2$ and $\rmax =0$.
		Active nodes are highlighted in bold.
}
	\label{fig:simfwdpush}
	\vspace{-4mm}
\end{figure}

In this section, we give a {\em positive} answer to the open question regarding to the running time of $\fwdpush$.
More specifically, we prove that under a {\em proper strategy} to pick active nodes to perform push operations,
the overall running time of $\fwdpush$ can be bounded by $O(m \cdot \log \frac{1}{\lmd})$ with $\rmax = \lmd / m$. 
This finding stems from an observation on a subtle equivalence connection between $\powitr$ and $\fwdpush$ as we discuss next.

\subsection{Equivalence  Connection to Power Iteration}
Recall that in each iteration, 
$\powitr$ essentially computes $\vga\itjnext = (1 - \alpha) \cdot \vga\itj \cdot \mP$.
Thus, the alive random walks considered in the same iteration are all with the same lengths.
Such a well structured process makes the error bound analysis of $\powitr$ relatively clear. 
In contrast, the process of $\fwdpush$ is a lot less structured. 
Due to the fact that $\fwdpush$ allows to perform push operations on arbitrary active nodes, 
the residues of the nodes actually mix up the probability mass of the alive random walk from $s$ at the states of different lengths. 
Despite of the similar rationale of moving alive random walks one-step forward in both of the algorithms, 
the arbitrary push operation ordering of $\fwdpush$ 
makes the analysis of the error bound during the algorithm very challenging. 
To overcome this challenge, we {\em conceptually} restrict $\fwdpush$ to perform push operations in iterations.

\vspace{1mm}
\noindent
{\bf A Special $\fwdpush$ Variant.}
As the first step, we reveal the subtle equivalence connection between $\powitr$ and $\fwdpush$.
In the following, we show a special variant of $\fwdpush$ which can perform exactly the same computation for $\vga\itjnext$ and $\pis\itjnext$ in $\powitr$. 
This variant is called {\em Simultaneous Forward Push} ($\simfwdpush$), and has the following modifications on Algorithm~\ref{algo:fwdpush}:
\begin{itemize}[leftmargin = *]
	\item All the nodes with {\em non-zero} residues are active, i.e., $\rmax = 0$.
	\item The $\simfwdpush$ algorithm works in iterations:
		\begin{itemize}
			\item At the beginning of the $(j+1)$-th iteration (for integer $j \geq  0$), the residue of node $v$ is denoted by $r\itj(s,v)$.
			\item In each iteration, the algorithm performs a push operation on every active node {\em simultaneously} based on $r\itj(s,v)$.
		\end{itemize}
	\item At the end of the $(j+1)$-th iteration, the algorithm terminates if the $\ell_1$-error $r_{sum} = \sum_{v\in V} r\itjnext(s, v) \leq \lmd$.
\end{itemize}

\noindent
{\bf A Running Example.} 
Figure~\ref{fig:simfwdpush} shows a running example of $\simfwdpush$.
At the beginning of the first iteration, only $v_1$ has non-zero residue, i.e.,  $r^{(0)}(s,v_1) = 1$, and thus, it is the only active node in this iteration.
After the push operation on $v_1$ 
$r^{(1)}(s,v_2) = r^{(1)}(s,v_3) = \frac{(1-\alpha) \cdot r^{(0)}(s,v_1)}{d_{v_1}}  = 0.4$.
Hence, $v_2$ and $v_3$ are the two active nodes in the second iteration. 
The algorithm then performs push operations {\em simultaneously} on both $v_2$ and $v_3$, 
where the operation on $v_2$ pushes $\frac{(1-\alpha) \cdot r^{(1)}(s,v_2)}{d_{v_2}} = 0.08$ probability mass to each $v_2$'s out-neighbor,
while the operation on $v_3$ pushes $\frac{(1-\alpha) \cdot r^{(1)}(s,v_3)}{d_{v_3}} =  0.16$ to its out-neighbors accordingly.
The resulted residue of each node is shown in the figure. 

\vspace{1mm}
\noindent
{\bf The Connection.}
Define $\vrs\itj \in \mathbb{R}^{1 \times n}$ as the {\em residue vector} of all the nodes, whose the $i$-th coordinate is $r\itj(s,v_i)$.
%
The crucial observation on $\simfwdpush$ is that performing simultaneous push operations on all the active nodes in the $(j + 1)$-th iteration is equivalent to the following computation:
{\small
\begin{align}\label{eq:simfwdpush}
	\vrs\itjnext & = \sum_{\text{active}\, v} \frac{(1 - \alpha) \cdot r\itj(s,v)}{d_{v}} \cdot \vec{A}_{v}
	= \sum_{v \in V} \frac{(1 - \alpha) \cdot r\itj(s,v)}{d_{v}} \cdot \vec{A}_{v} \nonumber \\
	&= \sum_{v \in V} (1-\alpha) \cdot r\itj(s, v) \cdot \vec{P}_{v} 
	= \sum_{v \in V} \left((1-\alpha) \cdot r\itj(s, v) \cdot \vec{e}_{v} \cdot \mP  \right) \nonumber \\
		     &= (1 - \alpha) \cdot \left(\sum_{v \in V} r\itj(s, v) \cdot \vec{e}_{v}\right) \cdot \mP 
		     	     = (1 - \alpha) \cdot \vrs\itj \cdot \mP \,. 
\end{align}
}%
We have the following lemmas.
\begin{lemma}\label{lmm:equivalent}
	The residue vector $\vrs\itjnext$ and underestmate PPR vector $\hpis\itjnext$ obtained by $\simfwdpush$ in the $(j+1)$-th iteration are exactly the same as $\vga\itjnext$ and $\vpisjnext$ computed in the $(j+1)$-th iteration 
in $\powitr$, for all integer $j \geq 0$.
\end{lemma}
\begin{proof}
	We prove this lemma with a mathematical induction argument. 
	Clearly, the base case $\vrs^{(0)} = \vga^{(0)} = \ve_s$  and $\hpis^{(0)} = \pis^{(0)} = \vec{0}$ holds.
	For the inductive case, assuming that $\vrs\itj = \vga\itj$  and $\hpis\itj = \pis\itj$ holds,
	by Equation~\eqref{eq:simfwdpush}, we have:
	{\small
	\begin{align*}
	\vrs\itjnext =  (1 - \alpha) \cdot \vrs\itj \cdot \mP =  (1 - \alpha) \cdot \vga\itj \cdot \mP = \vga\itjnext\,;
	\end{align*}
	}%
	and according to the push operations, 
	{\small
	\begin{align*}
	\hpis\itjnext = \hpis\itj + \alpha \cdot \vrs\itj = \sum_{k = 0}^{j} \alpha \cdot \vrs^{(k)} = \sum_{k = 0}^{j} \alpha \cdot \vga^{(k)} = \pis\itjnext\,.
	\end{align*}
	}%
	Therefore, the inductive case holds, and the lemma follows.
\end{proof}

\begin{lemma}\label{lmm:simtime}
	The overall running time of $\simfwdpush$ is bounded by $O(m \cdot \log \frac{1}{\lmd})$.
\end{lemma}
\begin{proof}
	The cost of each push operation on a node $v$ is $O(d_v)$.
	Thus, in each iteration, the total cost is bounded by the total degree $O(\sum_{v\in V} d_v) = O(m)$.
	According to the analysis of $\powitr$ and Lemma~\ref{lmm:equivalent}, after at most $O(\log \frac{1}{\lmd})$ iterations,
	the $\ell_1$-error $r_{sum}  \leq \lmd$. 
\end{proof}

Putting Lemmas~\ref{lmm:equivalent} and ~\ref{lmm:simtime} together,
we conclude that:
	$\simfwdpush$ is equivalent to $\powitr$.

\subsection{A Tighter Analysis}
Unfortunately, the equivalence between $\simfwdpush$ and $\powitr$ 
is not sufficient to answer the open question regarding to the running time of $\fwdpush$. 
The reasons are as follows:
\begin{itemize}[leftmargin = *]
	\item First, the push operations in $\simfwdpush$ are performed simultaneously in each iteration, 
		while in $\fwdpush$, they are performed in an {\em asynchronous} way.
	\item Second, the crucial parameter $\rmax$ does not make much effect in $\simfwdpush$, but it determines which node is eligible for a push operation in $\fwdpush$. 
	\item Third, the stop condition in $\simfwdpush$ that requires $r_{sum} \leq \lmd$ is not a sufficient condition to achieve
		$r(s,v) \leq d_v \cdot \rmax$ for all $v\in V$, where the latter is the original stop condition in $\fwdpush$.
\end{itemize}
In this subsection, we remove all these restrictions.
The only requirement in our analysis for $\fwdpush$ is that the algorithm is performed in iterations (just as what $\simfwdpush$ does and as defined below). 
We note that considering the algorithm in iterations makes the entire process more structured 
and thus allows us to bound the decrease rate of the $\ell_1$-error.
Nonetheless, as mentioned earlier, 
such a requirement is not a strong restriction; 
and it indeed can be implemented as simple as with a {\em First-In-First-Out} queue to organize the active nodes during the algorithm. 
Interestingly enough, this is actually a common implementation of $\fwdpush$ in practice -- 
people have {\em unconsciously} implemented $\fwdpush$ in an efficient way! 
From  our analysis, it explains why $\fwdpush$ is often found to have a weaker dependency on $\lmd$ than as what its previous running time complexity suggested in applications.

\begin{algorithm}[t]	
	\caption{First-In-First-Out Forward Push}
	\label{algo:itrfwdpush}
	\KwIn{$G$, $\alpha$, $s$, $\rmax$}
	\KwOut{an estimation $\hpis$ of $\pis$ and the resulted residues $\vrs$}
	$\hpi(s,v) \leftarrow 0$ and $r(s,v) \leftarrow 0$ for all $v \in V$;
	$r(s, s) \leftarrow 1$\;
	initialize a {\em first-in-first-out} queue $Q \leftarrow \emptyset$\;
	$Q.append(s)$; // append $s$ at the end of $Q$\;
	\While{$Q \neq \emptyset$}{
		$v \leftarrow Q.pop()$; // pop and remove the front node from $Q$\;
		$\hpi(s,v) \leftarrow \hpi(s,v) + \alpha \cdot r(s, v)$\;
		\For{each $u \in N_{out}(v)$}{
			$r(s, u) \leftarrow r(s,u) +  \frac{(1 - \alpha) \cdot r(s, v)}{d_v}$\;
			\If{$r(s,u) > d_u \cdot \rmax$ and $u \not\in Q$}{ 
				// $u$ is active and not in $Q$\;
				$Q.append(u)$; // append $u$ at the end of $Q$\;
			}
		}
		$r(s,v) \leftarrow 0$\; 
 	}
	\Return $\hpi(s,v)$ for all $v\in V$ as a vector $\hpis$, and $r(s,v)$ for all $v\in V$ as the resulted residue vector $\vrs$\;
\end{algorithm}

In the following, we analyse the running time of an implementation of $\fwdpush$, called {\em First-In-First-Out Forward Push} ($\itrfwdpush$), whose pseudo-code is shown in Algorithm~\ref{algo:itrfwdpush}.
We answer the open question regarding to the running time of $\fwdpush$ by proving the following theorem.
\begin{theorem}\label{thm:itrfwdpush}
	Given $0 < \rmax < \frac{1}{2\cdot m}$ and $\lmd = m \cdot \rmax$, the overall running time of $\itrfwdpush$ is bounded by $O(m \cdot \log \frac{1}{\lmd})$.
\end{theorem}
It should be noted that when $\rmax \geq \frac{1}{2 \cdot m}$, the bound $O(\frac{1}{\rmax}) = O(m)$ is already good enough. 
Furthermore, as aforementioned, the goal is to obtain high-precision results; the $\lmd$ value of interest is often even smaller than $10^{-8}$ 
and thus, $\rmax$ is often far smaller than $\frac{1}{2 \cdot m}$ in practice.

\vspace{1mm}
\noindent
{\bf The Iterations.} For the ease of analysis, we first define the iterations of $\itrfwdpush$ based on Algorithm~\ref{algo:itrfwdpush}.
In particular, we define $S\itj$ as the set of all the active nodes at the beginning of the $(j+1)$-th iteration, where $j = 0, 1, 2, \cdots$.
Specifically, we define $S\itj$ in an inductive way:
\begin{itemize}[leftmargin = *]
	\item Initially, $S^{(0)} = \{s\}$: the source node is the only active node at the beginning of the first iteration.
	\item $S\itjnext$ is the set of all the nodes appended to $Q$ at Line $11$ in Algorithm~\ref{algo:itrfwdpush} when processing the nodes in $S\itj$.
\end{itemize}
Furthermore, in the $(j+1)$-th iteration, the algorithm performs a push operation for every active node in $S^{(j)}$.

Under this definition, the iterations are exactly the same as those we considered in $\simfwdpush$, except that the push operations are now performed in an asynchronous way.
Consider the example in Figure~\ref{fig:simfwdpush} and assume $\rmax$ is sufficiently small, e.g., $0.001$; 
$S^{(0)}$ contains $v_1$ only.
After the push operation on $v_1$, only $v_2$ and $v_3$ are appended to $Q$; thus, $S^{(1)} = \{v_2, v_3\}$. 
In the second iteration, during the push operations on $v_2$ and $v_3$,
all the five nodes are appended to $Q$. 
Hence, $S^{(2)} = \{v_1, v_2, v_3, v_4, v_5\}$.

\vspace{1mm}
\noindent
{\bf An Overview of the  Analysis.} 
Let $\rsum\itj = \|\vrs\itj\|_1$, the total residues of all the nodes at the beginning of the $(j+1)$-th iteration. 
Initially, $\rsum^{(0)} = 1$.
According to the analysis for $\fwdpush$, 
we know that $\rsum\itjnext$ is exactly the $\ell_1$-error after the $(j+1)$-th iteration.
When the iteration number is not important, we use $\rsum$ to denote the $\ell_1$-error at the current state.

Our analysis on the overall running time of $\itrfwdpush$ consists of two main steps.
Firstly, we show that the following lemma:
\begin{lemma}\label{lmm:itrfwdpush-error}
	In $O(m \cdot \log \frac{1}{\lmd} + m)$ time, $\itrfwdpush$ can make the $\ell_1$-error $\rsum \leq \lmd$. 
\end{lemma}
As aforementioned, $\rsum \leq \lmd$ is not sufficient to guarantee that $r(s,v) \leq d_v \cdot \rmax$ holds for all $v\in V$.
Thus, the $\itrfwdpush$ algorithm may not stop and keep running until there is no more active node. 
To bound the running time of this part, in the second step, we prove:
\begin{lemma}\label{lmm:itrfwdpush-stoptime}
	Starting from the state of $\rsum \leq \lmd$, $\itrfwdpush$ stops in $O(m)$ time.
\end{lemma}

Theorem~\ref{thm:itrfwdpush} follows immediately from these two lemmas.
In the rest of this subsection, we prove Lemmas~\ref{lmm:itrfwdpush-error} and ~\ref{lmm:itrfwdpush-stoptime}, respectively.

\vspace{1mm}
\noindent
{\bf Proof of Lemma~\ref{lmm:itrfwdpush-error}.}
Consider the $(j+1)$-th iteration; $\itrfwdpush$ performs a push operation on each node $v \in S\itj$.
For each such push operation, an amount of $\alpha\cdot r\itj(s,v)$ probability mass is converted to $\hpi\itjnext(s,v)$, and hence,
$\rsum\itj$ is decreased by $\alpha \cdot r\itj(s,v)$.
Therefore, at the end of this $(j+1)$-th iteration,
the net decrease of $\rsum\itj$ is:
\begin{equation}\label{eq:rsum-error} \small
	\rsum\itj - \rsum\itjnext \geq  \alpha \cdot \sum_{v\in S\itj} r\itj(s,v)\,.
\end{equation}
The key in our proof is to show $\rsum\itjnext \leq (1 - \frac{\alpha}{m} \cdot \sum_{v \in S\itj} d_v) \cdot \rsum\itj$. 
To achieve this, we show the following observation.

\begin{observation}\label{obs1}
	$\sum_{v\in S\itj} r\itj(s,v) \geq \frac{1}{m} \cdot \sum_{v \in S\itj} d_v  \cdot \rsum\itj\,$.
\end{observation}
\begin{proof}
	In the following calculation, we omit all the superscripts in  $S\itj$, $\rsum\itj$ and the residues $r\itj(s,v)$ 
	as they are all with respect to $j$.
%
	Clearly, when $\sum_{v \in S} d_v = 0$ or $\sum_{v \in S} d_v = m$, the observation holds. Otherwise, by the definition of active nodes, we have:  
	{\small
	\begin{equation*}
		\frac{\sum_{v \in S} r(s, v) }{ \sum_{v \in S} d_v } \geq \rmax  \geq \frac{ \sum_{v \not \in S} r(s, v) }{ \sum_{v \not \in S} d_v } \,.
	\end{equation*}
	}%
	Therefore, it follows that:
	{\small
	\begin{equation*} 
	\frac{\sum_{v \in S} r(s, v) }{ \sum_{v \in S} d_v } \ge \frac{\sum_{v \in S} r(s, v) + \sum_{v \notin S} r(s, v) }{ \sum_{v \in S} d_v + \sum_{v \notin S} d_v } = \frac{r_{sum} }{ m}\,. 
	\end{equation*}
	}%
	The observation follows.
\end{proof}

Substituting Observation~\ref{obs1} to Equation~\eqref{eq:rsum-error}, we have:
{\small
\begin{align}\label{eq:rsum-error2}
	\rsum\itjnext &\leq \rsum\itj - \alpha \sum_{v \in S\itj} r\itj(s,v)
		      \leq (1 - \frac{\alpha}{m} \cdot \sum_{v \in S\itj} d_v) \cdot \rsum\itj \nonumber \\
		      &\leq \prod_{k = 0}^{j} (1 -  \frac{\alpha}{m} \cdot \sum_{v \in S^{(k)}} d_v) \cdot \rsum^{(0)} \nonumber \\
		      &\leq \exp\left(-\sum_{k = 0}^j (\frac{\alpha}{m} \cdot \sum_{v \in S^{(k)}} d_v) \right)
		      = \exp\left(-\frac{\alpha}{m} \cdot (\sum_{k = 0}^j \sum_{v \in S^{(k)}} d_v) \right)  \,,
\end{align}
}%
where the last inequality follows from the fact that $1 - x \leq e^{-x}$ holds for all $x \in \mathbb{R}$. 

Let $T\itjnext = \sum_{k = 0}^j \sum_{v \in S^{(k)}} d_v$ be the total degree of the node in each push operation performed in the first $(j+1)$ iterations. 
By Equation~\eqref{eq:rsum-error2}, in order to make $\rsum\itjnext \leq \lmd$, it suffices to find the {\em smallest} $j$ such that
{\small
\begin{align*}
\exp\left(-\frac{\alpha}{m} \cdot T\itjnext \right) \leq  \lmd \leq \exp\left(-\frac{\alpha}{m} \cdot T\itj \right) \,.
\end{align*}
}%
Thus, we have:
{\small
\begin{align*}
	T\itj \leq \frac{m}{\alpha} \cdot \ln \frac{1}{\lmd} \leq T\itjnext\,.
\end{align*}
}%
By the fact that $T\itjnext - T\itj = \sum_{v \in S\itj} d_v \leq m$, we further have:
{\small
\begin{align}
	T\itjnext \leq T\itj + m \leq \frac{m}{\alpha} \cdot \ln \frac{1}{\lmd} + m \,. 
\end{align}
}%
Finally, since the cost of a push operation on $v$ is bounded by $O(d_v)$, 
thus $O(T\itjnext)$ is actually an upper bound on the overall cost in the first $(j+1)$ iterations.
Therefore, the overall running time to achieve $\rsum\itjnext \leq \lmd$ is bounded by $O(m \cdot \log \frac{1}{\lmd} + m)$.
This completes the whole proof for Lemma~\ref{lmm:itrfwdpush-error}.

\vspace{1mm}
\noindent
{\bf Proof of Lemma~\ref{lmm:itrfwdpush-stoptime}.}
Let $\rbefore= \rsum\itjnext \leq \lmd$ be the $\rsum$ at the current state, and $\rafter$ be the $\rsum$ when the algorithm terminates.
Recall that each push operation on an active node $v$ decreases $\rsum$ by $\alpha \cdot r(s,v) \geq \alpha \cdot d_v \cdot \rmax$, and
the corresponding running time cost is $O(d_v)$.
Therefore, after paying a total running time cost of $O(T)$, the net decrease of $\rsum$ is at least $\alpha \cdot T \cdot \rmax$.
As the net decrease is at most $\rbefore - \rafter \leq \lmd$,
it follows that $\alpha \cdot T \cdot \rmax$ cannot be greater than $\lmd$.
Hence, $T \leq \frac{\lmd}{\alpha \cdot \rmax} = O(m)$.
Thus, the largest possible running time of $\itrfwdpush$ starting from the state of $\rsum \leq \lmd$ is bounded by $O(m)$. 
Lemma~\ref{lmm:itrfwdpush-stoptime} thus follows. 

%% file: powforpush.tex
\section{A New Efficient Power Iteration}\label{sec:powforpush}

In the previous section, we show that: (i) $\powitr$ is equivalent to a special variant of $\fwdpush$, 
and (ii) a simple implementation $\itrfwdpush$ of $\fwdpush$ can achieve time complexity $O(m \cdot \log \frac{1}{\lmd})$.
Based on these theoretical findings, in this section, we design an efficient implementation of $\powitr$, call {\em Power Iteration with Forward Push} ($\powforpush$), from an engineering point of view.
Our optimizations in the design of $\powforpush$ unifies the global-approach $\powitr$ and local-approach $\fwdpush$ and incorporates both their strengths.
Algorithm~\ref{algo:powforpush} is the pseudo-code of $\powforpush$. 
We introduce some crucial optimizations in $\powforpush$ in below.

\vspace{1mm}
\noindent
{\bf Asynchronous Pushes.}
Unlike $\powitr$,
our $\powforpush$ uses asynchronous push operations.
We note that asynchronous push operations can be possibly more effective.
This is because during the $(j+1)$-th iteration, if there is a push operation on an in-neighbor $u$ of a node $v$ before the push of $v$, 
when $v$ pushes, its current residue is greater than $r\itj(s,v)$, and hence, this push operation can send out more residue.
To see this, in the second iteration in Figure~\ref{fig:simfwdpush}, the simultaneous push operation on $v_2$ is performed based on a residue of $0.4$
but in the same iteration in Figure~\ref{fig:fwdpush}, the push on $v_2$ is based on a residue of $0.56$. 
This is because $v_3$ pushed before $v_2$, and hence, $v_2$'s residue has been increased by $0.16$. 
Moreover, after this asynchronous push, the residue of $v_2$ becomes $0$ in the next iteration, while in contrast, $v_2$ still has $0.16$ (obtained from $v_3$) under the simultaneous pushes. 
In other words, this asynchronous push on $v_2$ has equivalently processed the residues of $v_2$ in two iterations under simultaneous pushes.

\vspace{1mm}
\noindent
{\bf Global Sequential Scan v.s. Local Random Access.}
One of the biggest optimizations in $\powforpush$ is the strategy of switching to a global sequential scan from using the queue to access active nodes.
The key observation is that after a few iterations, in $\itrfwdpush$, there would be a large number of active nodes which are stored in the queue according to their ``append-to-queue'' order.
As a result, to perform push operations on these nodes, it requires a large number of {\em random access} in both the node list and the edge list, incurring a substantial overhead.

To remedy this, in $\powforpush$, when the current number of active nodes is greater than a specified $scanThreshold$, it switches to sequential scan the node list to perform push operations on the active nodes (as shown in Algorithm~\ref{algo:powforpush} Line $15$ - $24$).
Moreover, to further facilitate this idea, $\powforpush$ stores all the nodes sorted by id's and concatenates the adjacent lists of the nodes in the same order (i.e., sorted by id's) in a large array.
Thanks to this storage format, in each iteration, $\powforpush$ can perform push operations on active nodes via a sequential scan on this edge array, which in turn has largely make the memory access patterns become cache-friendly.
Interestingly, this idea is borrowed from the implementation of $\powitr$ as a global-approach.

\vspace{1mm}
\noindent
{\bf Dynamic $\ell_1$-Error Threshold.}
Another optimization worth mentioning is the strategy of using dynamic $\ell_1$-error threshold (see Line $17$ in Algorithm~\ref{algo:powforpush}).
The rationale here is that with a larger $\ell_1$-error threshold, it allows us to use a larger $\rmax$. 
We note that $\rmax$ essentially specifies a threshold on the {\em unit-cost benefit} of the push operations.
To see this, recall that a push operation on $v$ takes $O(d_v)$ cost and reduces $\rsum$ by $\alpha \cdot r(s,v)$.
Thus, $\alpha \cdot r(s,v) / d_v$ can be considered as the unit-cost benefit of this operation. 
By definition, a node becomes active only if a push operation on it has unit-cost benefit $\geq \alpha \cdot \rmax$.
The good thing of performing push operations with higher unit-cost benefits first is that 
it allows other nodes to accumulate their residues before pushing.
In this way, the number of push operations to achieve $\ell_1$-error can be considerably reduced. 
Motivated by this, we perform $\powforpush$ in epochs. 
In the $i$-th ($1 \leq i \leq epochNum$) epoch, 
an $\ell_1$-error $\lmd^{\frac{i}{epochNum}} \geq \lmd$ is adopted to perform those push operations with higher unit-cost benefits.


\vspace{1mm}
\noindent
{\bf Remark.}
The $\hpis$ and $\vrs$ returned by $\powforpush$ can be further refined to ensure $r(s,v) \leq d_v \cdot \rmax$, where $\rmax = \frac{\lmd}{m}$, holds for all $v\in V$.
By Lemma~\ref{lmm:itrfwdpush-stoptime}, such refinement only takes $O(m)$ time. 

\begin{algorithm}[t]
	\caption{Power Iteration with Forward Push}
	\label{algo:powforpush}
	\KwIn{$G$, $\alpha$, $s$, $\lmd$}
	\KwOut{an estimation $\hpis$ of $\pis$ and the resulted residues $\vrs$}
	$epochNum = 8$; // a tunable constant\;
	$scanThreshold = n / 4$; // a threshold to use sequential scan\;
	$\hpi(s,v) \leftarrow 0$ and $r(s,v) \leftarrow 0$ for all $v \in V$;
	$r(s, s) \leftarrow 1$\;
	initialize a {\em first-in-first-out} queue $Q \leftarrow \emptyset$\;
	$Q.append(s)$\; 
	$\rmax \leftarrow \lmd/m$\;
	\While{$Q \neq \emptyset$ and $Q$.size() $\leq$ $scanThreshold$ and $\rsum > \lmd$}{
			$v \leftarrow Q.pop()$\; 
			$\hpi(s,v) \leftarrow \hpi(s,v) + \alpha \cdot r(s, v)$\;
			\For{each $u \in N_{out}(v)$}{
				$r(s, u) \leftarrow r(s,u) +  \frac{(1 - \alpha) \cdot r(s, v)}{d_v}$\;
				if $u\not \in Q$ is active w.r.t. $\rmax$, $Q.append(u)$\;
			}
				$r(s,v) \leftarrow 0$\; 
	}
	\If{$\rsum > \lmd$}{
	// Switch to using sequential scan\;
	\For{$i \leftarrow 1; i \leq epochNum; i \leftarrow i + 1$}{
		$\rmax' \leftarrow \lmd^{\frac{i}{epochNum}} / m$; // allow a larger $\ell_1$-error\;
		\While{$\rsum > m \cdot \rmax'$}{	
			\For{each $v\in V$}{
				\If{$v$ is active w.r.t. $\rmax'$}{		
					$\hpi(s,v) \leftarrow \hpi(s,v) + \alpha \cdot r(s, v)$\;
					\For{each $u \in N_{out}(v)$}{
						$r(s, u) \leftarrow r(s,u) +  \frac{(1 - \alpha) \cdot r(s, v)}{d_v}$\;
					}
					$r(s,v) \leftarrow 0$\;
					}
				}
			}
		}
	}
	\Return $\hpi(s,v)$ for all $v\in V$ as a vector $\hpis$, and $r(s,v)$ for all $v\in V$ as the resulted residue vector $\vrs$\;
\end{algorithm}

%% file: speedppr.tex
\section{Improved Approx-SSPPR Algorithm}\label{sec:speedppr}
In this section, we propose a new algorithm, called $\speedppr$, for answering approximate SSPPR queries.

\subsection{Preliminaries on Approx-SSPPR}
In this subsection, we first introduce some preliminaries on two relevant algorithms:
$\mc$  and $\fora$.
The ideas of these two algorithms would help understand the key idea of the design of our $\speedppr$.
Recall that an Approx-SSPPR query aims to compute an estimation $\hpi(s,v)$ for every node $v$ with $\pi(s,v) \geq \mu$ within {\em relative error} $\eps$ with a succeed probability at least $1 - \frac{1}{n}$.

\vspace{1mm}
\noindent
{\bf The Monte Carlo Method.}
Perhaps, one of the most straightforward ways to answer Approx-SSPPR query is the $\mc$ method.
The basic idea is to generate $W$ independent $\alpha$-random walks from $s$, and utilise the {\em empirical} number $f(s,v)$ out of these random walks that stop at a node $v$ to estimate its expectation $\pi(s,v) \cdot W$.
Thus, $\hpi(s,v) = \frac{f(s,v)}{W}$ gives an estimation of $\pi(s,v)$.
By the standard Chernoff Bound~\cite{chung2006}, it is known that setting
{\small
\begin{align}\label{eq:W}
	W = \frac{2 \cdot (2\cdot \eps/ 3 + 2) \cdot \log n}{\eps^2 \cdot \mu} = O(\frac{\log n}{\eps^2 \cdot \mu})
\end{align}
}%
suffices to obtain a correct estimation for every node $v$ with $\pi(s,v) \geq \mu$ with probability at least $1 - \frac{1}{n}$.  
Furthermore, as the expected length of an $\alpha$-random walk is at most $\frac{1}{\alpha}$, 
the overall expected running time of $\mc$ is bounded by $O(\frac{\log n}{\eps^2 \cdot \mu})$.
When $\mu = \frac{1}{n}$, this bound can be written as $O(\frac{n \cdot \log n}{\eps^2})$. 

In the rest of this section, without loss of generality,  we assume that $m < W$, because otherwise, i.e., $m \geq W$, 
one can always switch their algorithm to the $\mc$ method and guarantee a time complexity $O(W)$ no worse than $O(m)$. 

\vspace{1mm}
\noindent
{\bf FORA.}
$\fora$~\cite{WYXWY17} is a state-of-the-art representative algorithm for answering Approx-SSPPR queries.
It adopts a two-phase framework and combines $\fwdpush$ and $\mc$.
In the first phase, it runs $\fwdpush$ with a specified $\rmax$ (whose value is to be determined shortly)
to obtain an estimation $\hpis$ of $\pis$ with $\abserror = \rsum = \sum_{v \in V} r(s,v)  \leq m \cdot \rmax$. 
In the second phase, it performs the $\mc$ method. 
Specifically, it works as follows.
For each node $v$ with $r(s,v) > 0$, $\fora$ generates $W_v = \lceil r(s,v) \cdot W \rceil$ random walks from $v$,
where $W$ is set by Equation~\eqref{eq:W}. 
Among these $W_v$ random walks from $v$, if $f(v,u)$ out of them had stopped at a node $u$, then increase $\hpi(s,u)$ by:
{\small
\begin{align}\label{eq:fora1}
	\frac{r(s,v) \cdot W}{W_v} \cdot \frac{f(v,u)}{W} = r(s,v) \cdot \frac{f(v,u)}{W_v}\,.
\end{align}
}%
In summary, the final estimation $\hpi'(s,u)$ of $\pi(s,u)$ is computed as: 
{\small
\begin{equation}\label{eq:fora2}
	\hpi'(s,u) = \hpi(s,u) + \sum_{v\in V : r(s,v) > 0} r(s,v) \cdot \frac{f(v,u)}{W_v}\,,
\end{equation}
}%
where $\hpi(s,u)$ is obtained in the first $\fwdpush$ phase, and the second term is the net increase based on $\hpi(s,u)$ in the $\mc$ phase.

\vspace{1mm}
\noindent
{\bf Running Time Analysis.}
According to the previous bound on the running time of $\fwdpush$, the cost of the first phase in $\fora$ is bounded by $O(\frac{1}{\rmax})$;
and in the second phase, $\fora$ needs to generate at most $\sum_{v \in V: r(s,v) >0} 
\lceil r(s,v) \cdot W \rceil \leq \rsum \cdot W + n \leq m \cdot \rmax \cdot W + n$ random walks. 
Therefore, the overall expected time of $\fora$ is bounded by $O(\frac{1}{\rmax} + m \cdot \rmax \cdot W + n)$, 
which can be minimized to $O(\sqrt{m \cdot W} + n)$ by setting $\rmax = \frac{1}{\sqrt{m \cdot W}}$.
When $\mu = \frac{1}{n}$ and the graph is {\em scale-free}, i.e., $m = O(n\cdot \log n)$,
this bound can be further simplified to $O(\frac{n \cdot \log n}{\eps})$. 
In this case, $\fora$ improves the $\mc$ method by a factor of $\frac{1}{\eps}$, 
Furthermore, this $O(\frac{n \cdot \log n}{\eps})$-bound is actually state-of-the-art; 
none of the existing algorithms can overcome this barrier.

\vspace{1mm}
\noindent
{\bf Pre-Computing the Random Walks.}
An optimization of $\fora$ is to {\em pre-compute} $K_v$ random walks for each node $v\in V$, where $K_v = d_v \cdot \sqrt{\frac{W}{m}} + 1 \geq W_v$;
when answering a query, it just needs to read the pre-computed random walk results to perform the second $\mc$ phase.
Therefore, the actual query cost can be further reduced. 
Such an index-based variant is called $\fora$+.
The space consumption of all the pre-computed random walk results is  $\sum_{v\in V} K_v = \sqrt{m \cdot W} + n$.
When $ m = O(n\cdot \log n)$, this gives the overall space consumption bound $O(\frac{1}{\eps} \cdot n \cdot \log n)$.
Unfortunately, as the number of pre-computed random walks for each node depends on $W$ and hence on the relative error $\eps$,
the index of $\fora$+ constructed for an $\eps = \eps_1$ is not sufficient to answer queries with relative error $\eps_2 < \eps_1$.
Moreover, to support queries with small $\eps$, the index requires a substantial space consumption.
These drawbacks have significantly limited the applicability of $\fora$+.

\subsection{Our Improved Algorithm}
Next, we propose a new Approx-SSPPR algorithm, called $\speedppr$, which not only improves $\fora$'s running time complexity, but also admits an index with size independent to $\eps$.  
While it eventually turns out that $\speedppr$ is as simple as substituting $\powforpush$ along with a $O(m)$-time post-refinement (to ensure that no node is active with respect to $\rmax = \frac{1}{W}$) in 
the first phase of $\fora$, 
it is our new $\powforpush$ technique to make these improvements of $\speedppr$ over $\fora$ become possible. 
The pseudo-code of $\speedppr$ is shown in Algorithm~\ref{algo:speedppr}.

\begin{theorem}\label{thm:speedppr-time}
	The overall expected running time of $\speedppr$ is bounded by 
	$O(m \cdot \log \frac{W}{m})$, where $W$ is computed as Equation~\eqref{eq:W}. 
	When the graph is {\em scale-free}, i.e., $m = O(n\cdot \log n)$ and $\mu = \frac{1}{n}$, this bound can be written to $O(n \cdot \log n \cdot \log \frac{1}{\eps})$. 
\end{theorem}
\begin{proof}
The correctness of $\speedppr$ follows immediately from $\fora$. It thus suffices to bound the expected running time. 


In the first phase, the cost of running $\powforpush$ with $\rmax = \frac{1}{W}$ is bounded by $O(m \cdot \log \frac{W}{m} + m)$.
In the second phase, for each node $v$ with $r(s,v) > 0$, $\speedppr$ needs to perform $W_v = \lceil r(s,v) \cdot W \rceil \leq \lceil d_v \cdot \rmax \cdot W \rceil = d_v$ random walks. Thus, in total, there are at most $m$ random walks needed, and hence, the expected running time for performing them is $O(m)$.
Putting the two cost together, the overall expected running time of $\speedppr$ is bounded by $O(m \cdot \log \frac{W}{m})$.
Furthermore, when $m = O(n\cdot \log n)$ and $\mu = \frac{1}{n}$, the bound is simplified to $O(n \cdot \log n \cdot \log \frac{1}{\eps})$.
\end{proof}

\vspace{-2mm}
\noindent
{\bf Improvements over $\fora$.} 
Despite of the analogous algorithm framework, $\speedppr$ has two significant improvements over $\fora$.
\begin{itemize}[leftmargin = *]
	\item First, the overall expected running time of $\speedppr$ improves $\fora$'s state-of-the-art $O(\frac{1}{\eps} \cdot n \cdot \log n)$-bound by {\em almost} a factor of $\frac{1}{\eps}$.
		Given the importance of Approx-SSPPR queries, our improved $\speedppr$ not only reduces the computational cost of the tasks, but also offers an opportunity for users to obtain more accurate results (by setting $\eps$ smaller)  with the same running time budget.  
	\item Second, in the $\mc$ phase of $\speedppr$, only at most $d_v$ random walks are needed for each node $v \in V$.
		As a result, an index with at most $m$ pre-computed random walk results suffices to support $\speedppr$ to answer any Approx-SSPPR queries with any $\eps$.	
		In contrast, as aforementioned, the index size of $\fora$+ depends on $\eps$.
		The index of $\speedppr$ can consume an order-of-magnitude less space than that of $\fora$ when $\eps$ is small.
		More importantly, $\speedppr$ has no need to re-build the index for different $\eps$'s. 
\end{itemize}

\begin{algorithm}[t]	
	\caption{$\speedppr$}
	\label{algo:speedppr}
	\KwIn{$G$, $\alpha$, $s$, $\eps$, $\mu$}
	\KwOut{an estimation $\hpis$ of $\pis$}
	$W \leftarrow  \frac{2 \cdot (2\cdot \eps/ 3 + 2) \cdot \log n}{\eps^2 \cdot \mu}$\;
	$\hpis, \vrs \leftarrow$ invoke $\powforpush$ with $G$, $\alpha$, $s$ and $\lmd = \frac{m}{W}$\;
	refine $\hpis$ and $\vrs$ to ensure no node is active w.r.t. $\rmax = \frac{1}{W}$\; 
	\For{each $v \in V$ with $r(s,v) > 0$}{
		$W_v \leftarrow \lceil r(s,v) \cdot W \rceil$\;
		perform $W_v$ random walks from $v$\;
		\For{each walk stopping at a node $u$}{
			$\hpi(s, u) \leftarrow \hpi(s, u) + \frac{r(s,v)}{W_v}$\;
		}
	}
	\Return $\hpi(s,v)$ for all $v\in V$ as a vector $\hpis$\;
\end{algorithm}

%% file: related-work.tex

\section{Other Related Work} \label{sec:related-work}

Single-source Personalized PageRank queries have been extensively studied for the past decades~\cite{JehW03,jung2017bepi,WuJZ14,WYXWY17,coskun2016efficient,AndersenBCHMT07,AndersenCL06,FujiwaraNSMO13, FujiwaraNSMO13sigmod,FujiwaraNYSO12,YuM16,FujiwaraNOK12,YuL13,fogaras2005towards,BackstromL11,SarmaMPU15,maehara2014computing,ZhuFCY13,ShinJSK15,BahmaniCX11,Chakrabarti07, GuoCCLL17,Ren2014clude,BahmaniCG10,OhsakaMK15,ZhangLG16,GuptaPC08,lofgren2014fast,lofgren2015personalized,LofgrenBG15}.
Among these works, \cite{page1999pagerank,maehara2014computing,ZhuFCY13,ShinJSK15,jung2017bepi,Chakrabarti07} consider exact SSPPR queries, which is most relevant to our work. 
The vanilla $\powitr$  algorithm is proposed in \cite{page1999pagerank} to compute high precision results of SSPPR queries. 
\cite{maehara2014computing} improves the efficiency of $\powitr$ by introducing a core-tree decomposition. 
BEAR~\cite{ShinJSK15} preprocess the adjacency matrix so that it contains a large and easy-to-invert submatrix,  and precomputes several matrices required for inverting the submatrix to form an index. 
BePI~\cite{jung2017bepi} is the state-of-the-art matrix-based index-oriented algorithm for computing the exact values of SSPPR. 
Like BEAR, BePI achieves high efficiency by precomputing several matrices required by $\powitr$ algorithm and storing them as an index. 
BePI improves over BEAR by employing $\powitr$ instead of matrix inversion, which avoids the $O(n^3)$ complexity.
However, the index size of BePI and BEAR could exceed the graph size by orders of magnitude, which limits their scalability on large graphs. 

There are also several methods \cite{ZhuFCY13,maehara2014computing,lofgren2014fast,lofgren2015personalized,LofgrenBG15,WYXWY17,WangTXYL16} for approximate SSPPR queries. Among them, {\em BiPPR}~\cite{lofgren2015personalized} combines Backward Search with the Monte-Carlo method to obtain a more accurate estimation for SSPPR. {\em HubPPR}~\cite{WangTXYL16} precomputes Forward and Backward Search results for "hub" nodes to speed up the PPR computation. 
{\em FORA}~\cite{WYXWY17} combines Forward Search with the Monte-Carlo method, which avoids performing Backward Search on each node in the graph.  
ResAcc~\cite{lin2020index} accelerates $\fora$ by accumulating the residues that returned to the source node in the $\fwdpush$ phase and ``distribute'' this residue to other nodes proportionally based on $\hpis$ prior to the Monte-Carlo phase. 

Another line of research on PPR focuses on top-$k$ PPR queries~\cite{FujiwaraNOK12,FujiwaraNYSO12,FujiwaraNSMO13,FujiwaraNSMO13sigmod,YuL13,WuJZ14,coskun2016efficient}. Local update based methods~\cite{FujiwaraNOK12,FujiwaraNYSO12,FujiwaraNSMO13,FujiwaraNSMO13sigmod,YuL13,WuJZ14} performs a local search from the source node $s$ while maintaining lower and upper bounds of each node's PPR, and stops the search once the lower and upper bounds give the top-$k$ results.  For example, \cite{coskun2016efficient} improves Power Iteration by utilizes Chebyshev polynomials for acceleration. TopPPR~\cite{wei2018topppr} combines Forward Search, Backward Search, and the Monte-Carlo method to obtain exact top-$k$ results. These methods focus on refining the lower and upper bounds of the top-$k$ PPR values and thus are orthogonal to the techniques discussed in this paper.

\begin{table}[t]
    \begin{tabular}{|l|c|c|c|c|}
    \hline
    Name             & n     & m     & m/n  & Type       \\ \hline
    \textit{DBLP}    & 317K  & 2.10M & 6.62 & undirected \\ \hline
    \textit{Web-St}  & 282K  & 2.31M & 8.20 & directed   \\ \hline
    \textit{Pokec}   & 1.63M & 30.6M & 18.8 & directed   \\ \hline
    \textit{LJ}      & 4.85M & 68.4M & 14.1 & directed   \\ \hline
    \textit{Orkut}   & 3.07M & 234M  & 76.3 & undirected \\ \hline
    \textit{Twitter} & 41.7M & 1.47B & 35.3 & directed   \\ \hline
    \end{tabular}
    \caption{The Six Datasets ($K = 10^3, M = 10^6, B = 10^9$) }
    \label{tab:datasets}
    \vspace{-8mm}
\end{table}

%% file: experiment.tex
 \begin{figure*}{t}
	 \vspace{-2mm}
	 \hspace{-8mm}
   \includegraphics[width=.8\linewidth]{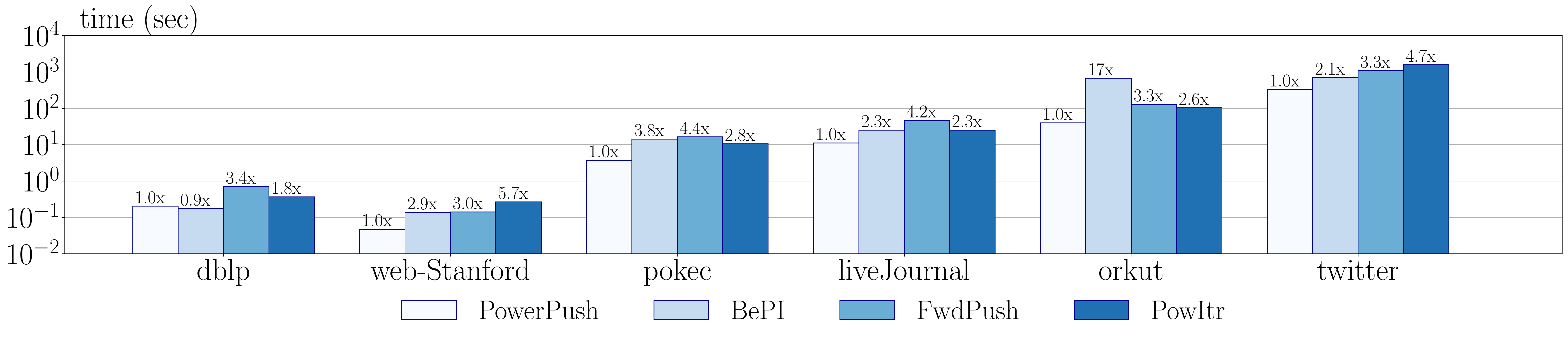} \\
	\vspace{-1mm}
     \figcapup 
     \caption{running time v.s. datasets, where the number $c.c \times$ over each bar means $c.c$ times of $\powforpush$'s running time } 
     \label{fig:time_vs_datasets}
 	\vspace{-3mm}
 \end{figure*}

  \begin{figure*}
   \includegraphics[width=0.5\linewidth]{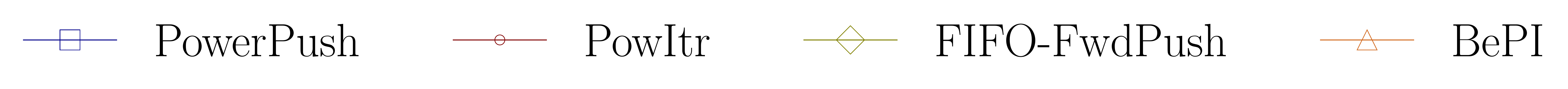} \\
\resizebox{.9\linewidth}{!}{
	 \begin{tabular}{cccccc} 
	\hspace{-6mm} 
    \includegraphics[width=0.16\linewidth]{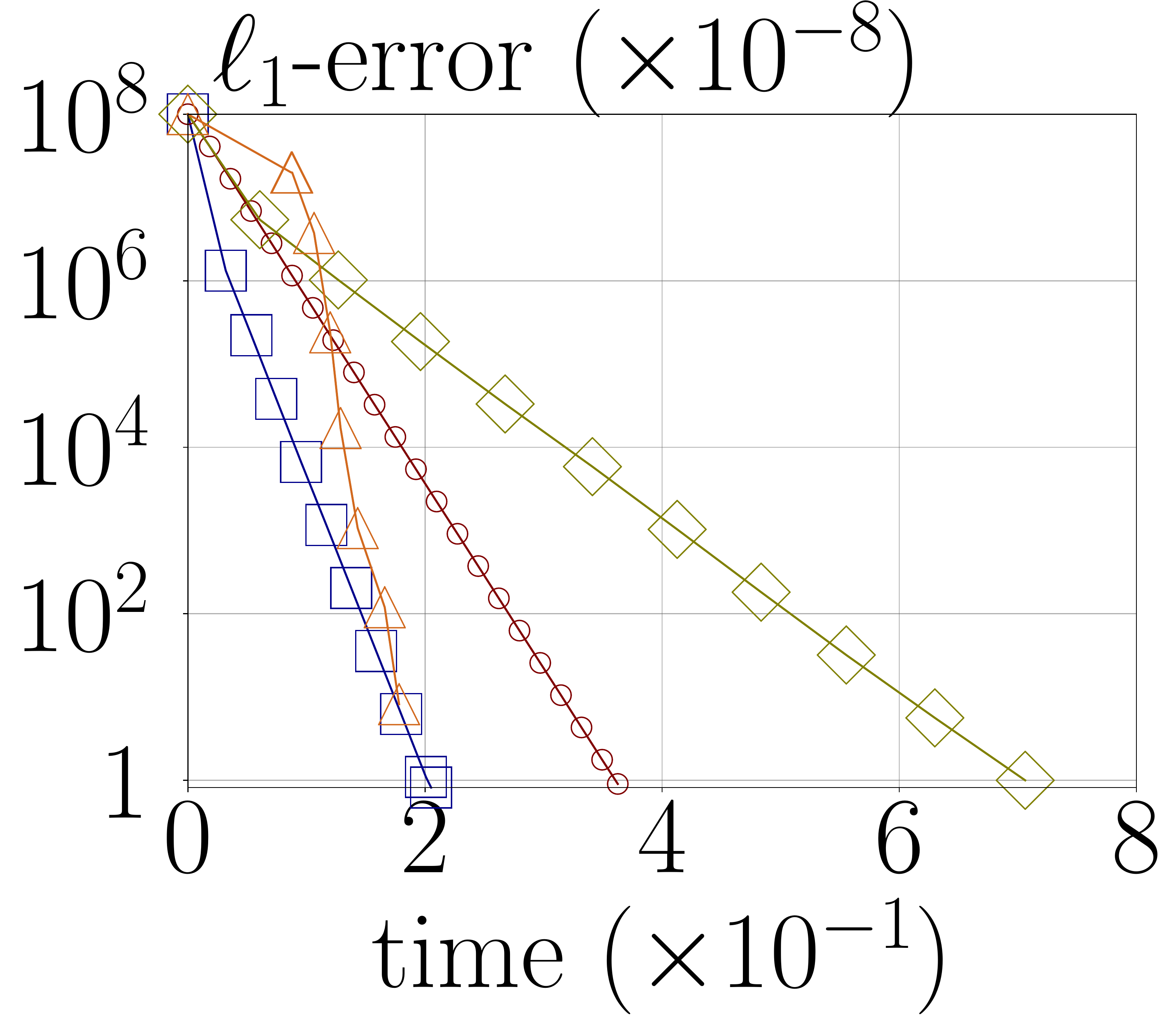} & 
   \hspace{-4.5mm}
	\includegraphics[width=0.16\linewidth]{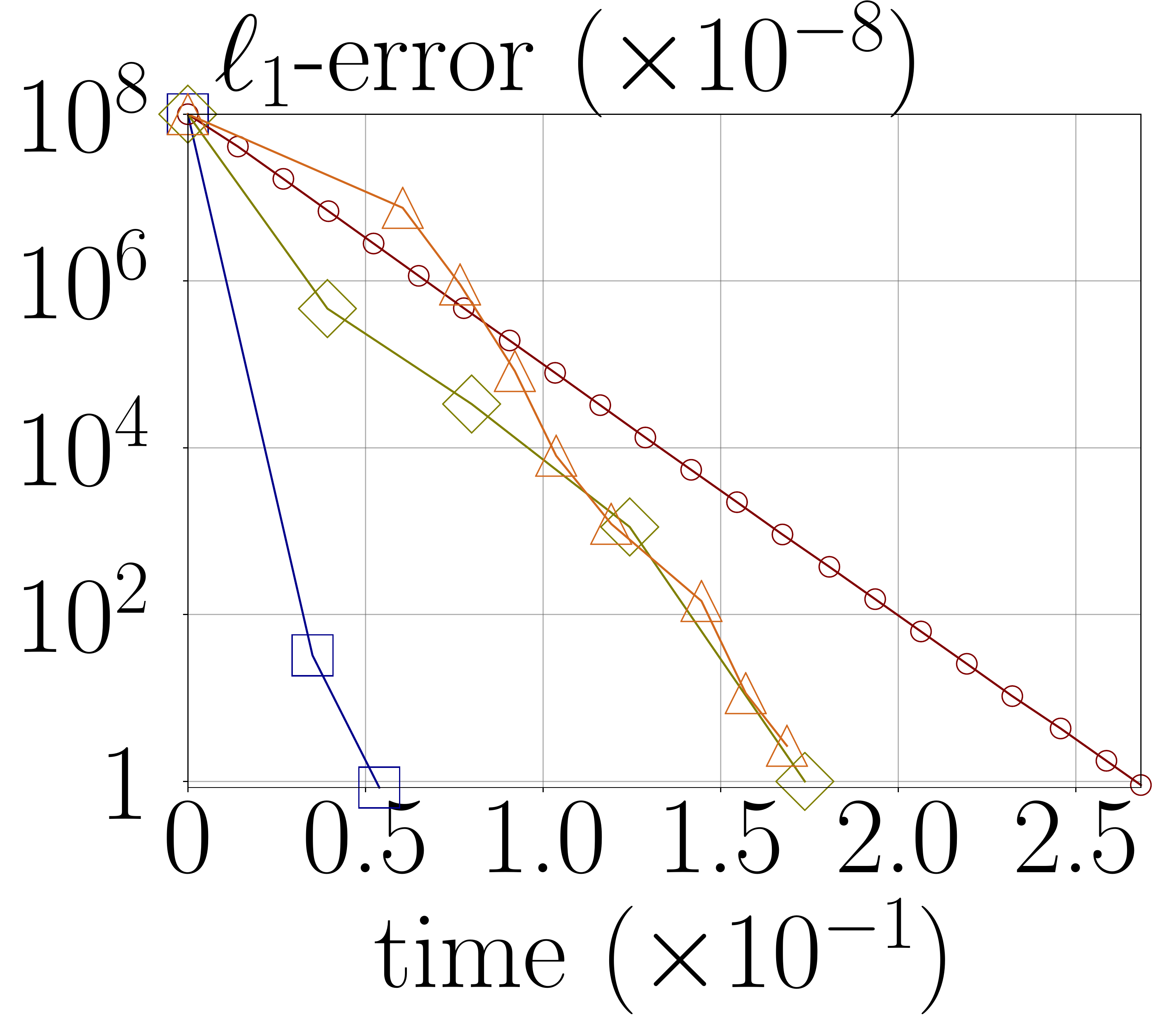} & 
   \hspace{-4.5mm}
	\includegraphics[width=0.16\linewidth]{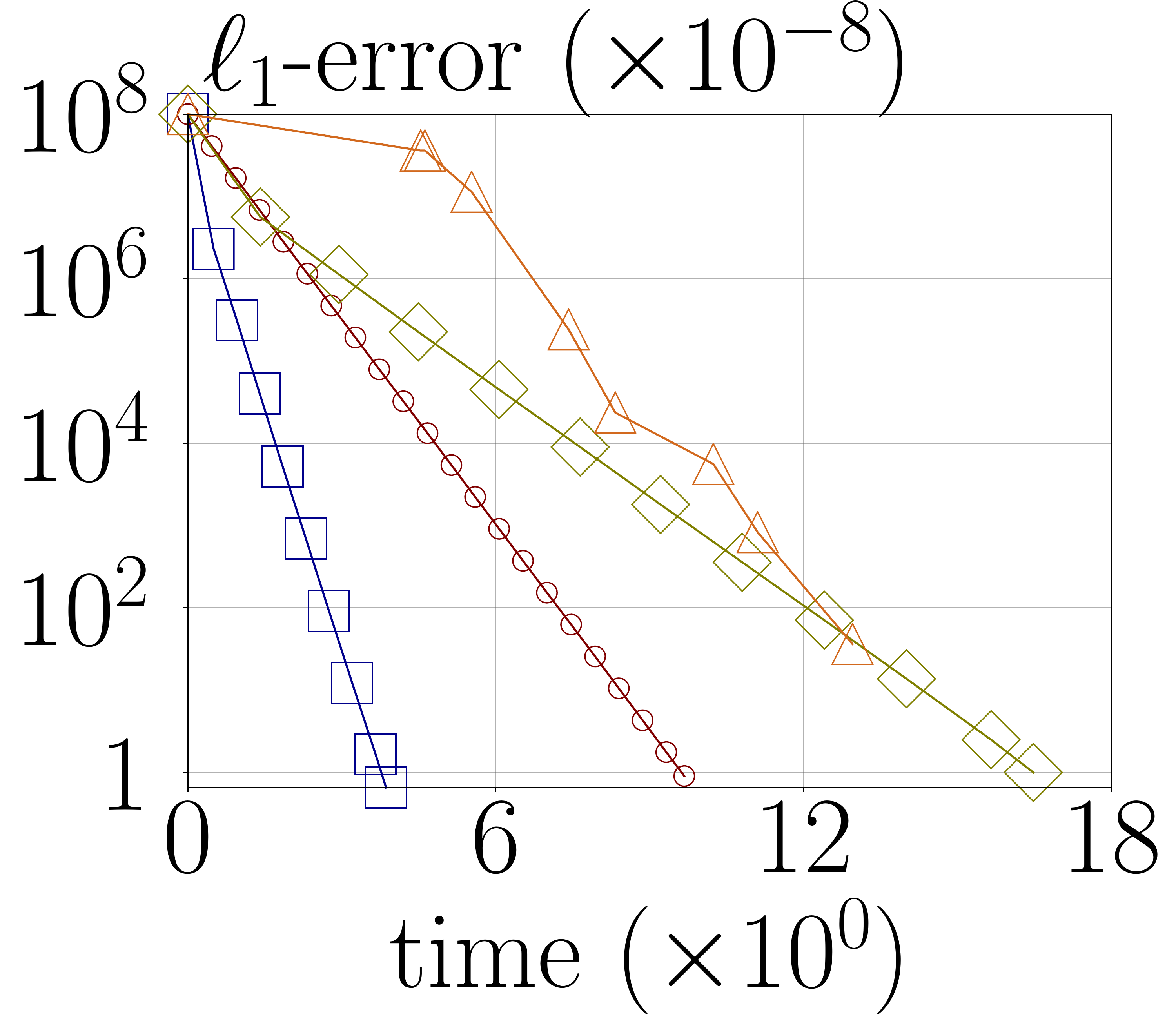} & 
   \hspace{-4.5mm}
	\includegraphics[width=0.16\linewidth]{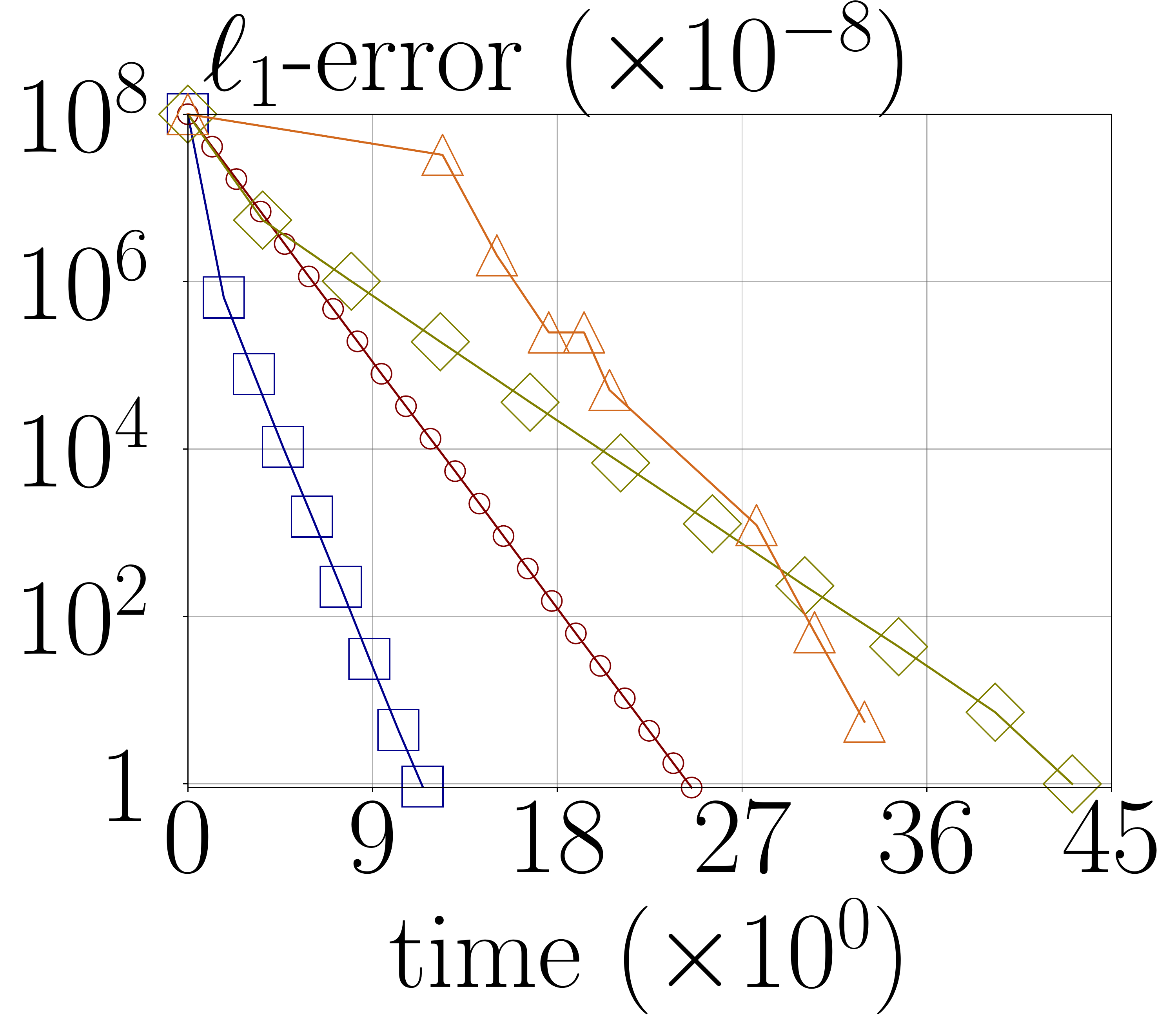} & 
   \hspace{-4.5mm}
	\includegraphics[width=0.16\linewidth]{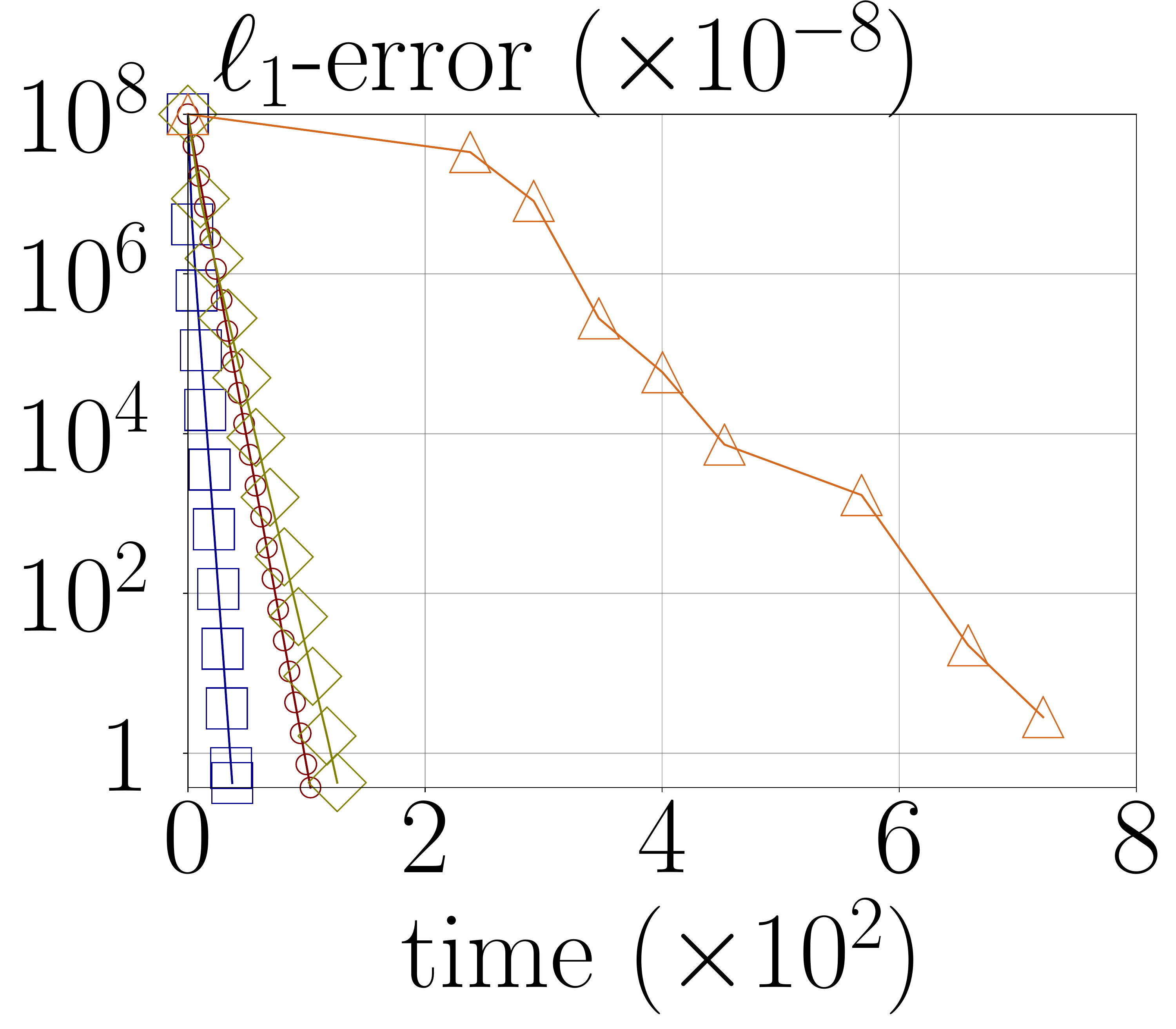} & 
   \hspace{-4.5mm}
   \includegraphics[width=0.16\linewidth]{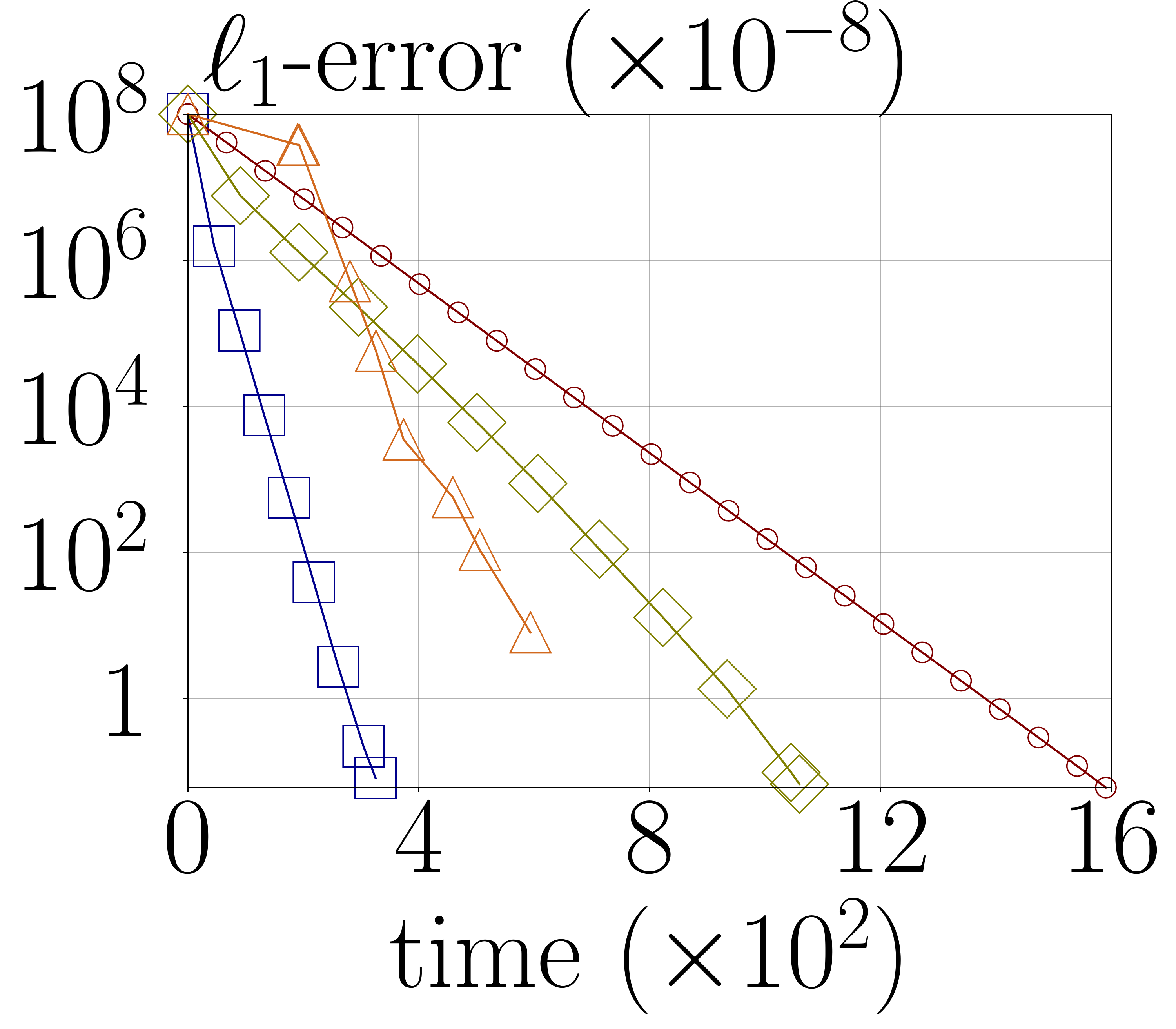} \\[-2mm]
	
	\hspace{-10mm} (a) {\em dblp} & \hspace{-4.6mm} (b) {\em web-Stanford} & \hspace{-4.6mm} (c) {\em pokec}  & \hspace{-4.6mm} (d)  {\em liveJournal}  & \hspace{-4.6mm}  (e)  {\em orkut}  & \hspace{-4.6mm} (f)  {\em twitter}  \\[1mm]
	
	\end{tabular}
}
	\vspace{-5mm}
     \caption{actual $\ell_1$-error v.s. execution time (seconds)} 
     \label{fig:time_vs_r_sum}
 	\vspace{-2mm}
 \end{figure*}

 \begin{figure*}
	 \vspace{-2mm}
\resizebox{.9\linewidth}{!}{
	 \begin{tabular}{cccccc} 
	\hspace{-6mm} 
    \includegraphics[width=0.18\linewidth]{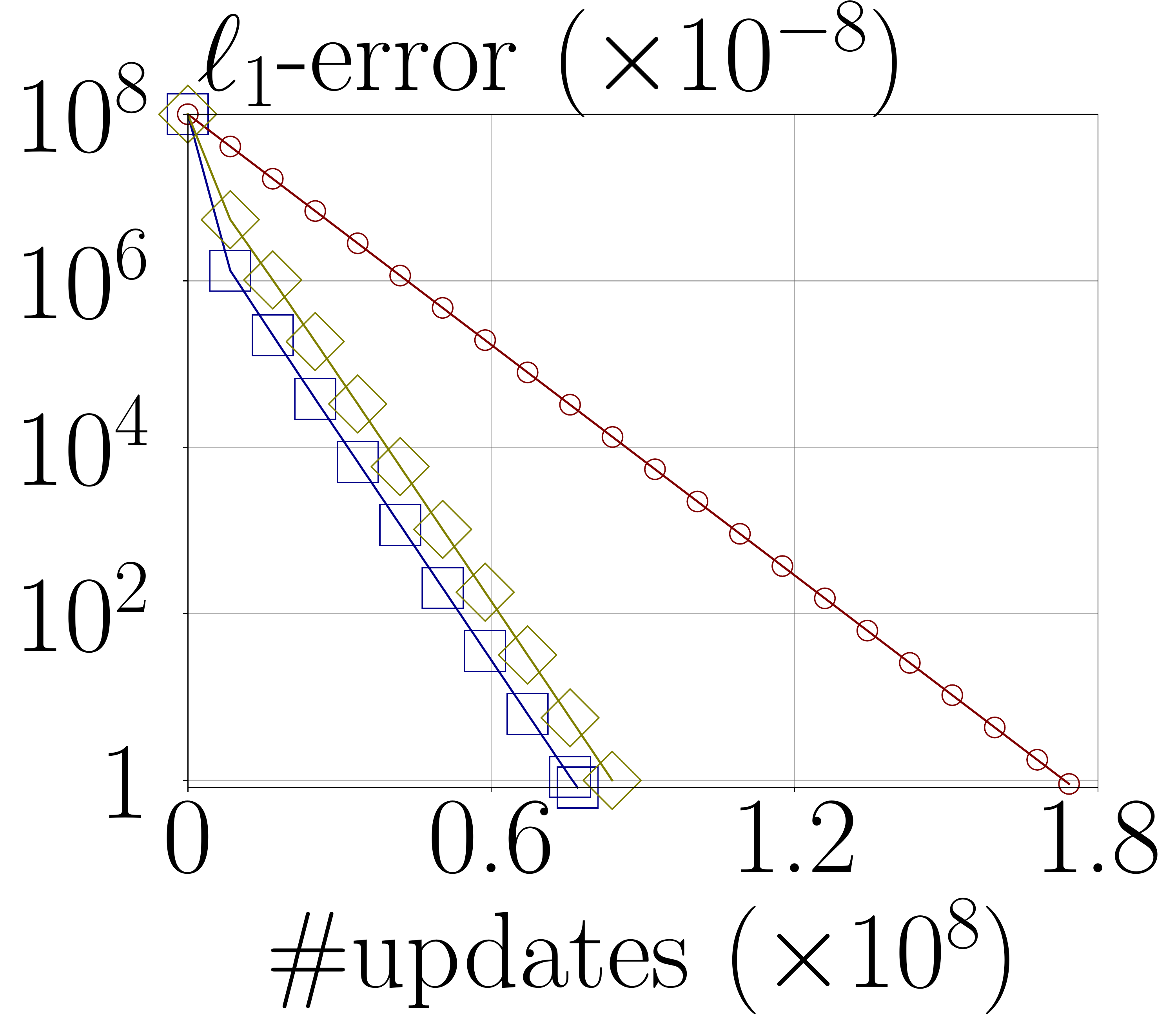} & 
   \hspace{-4.5mm}
	\includegraphics[width=0.18\linewidth]{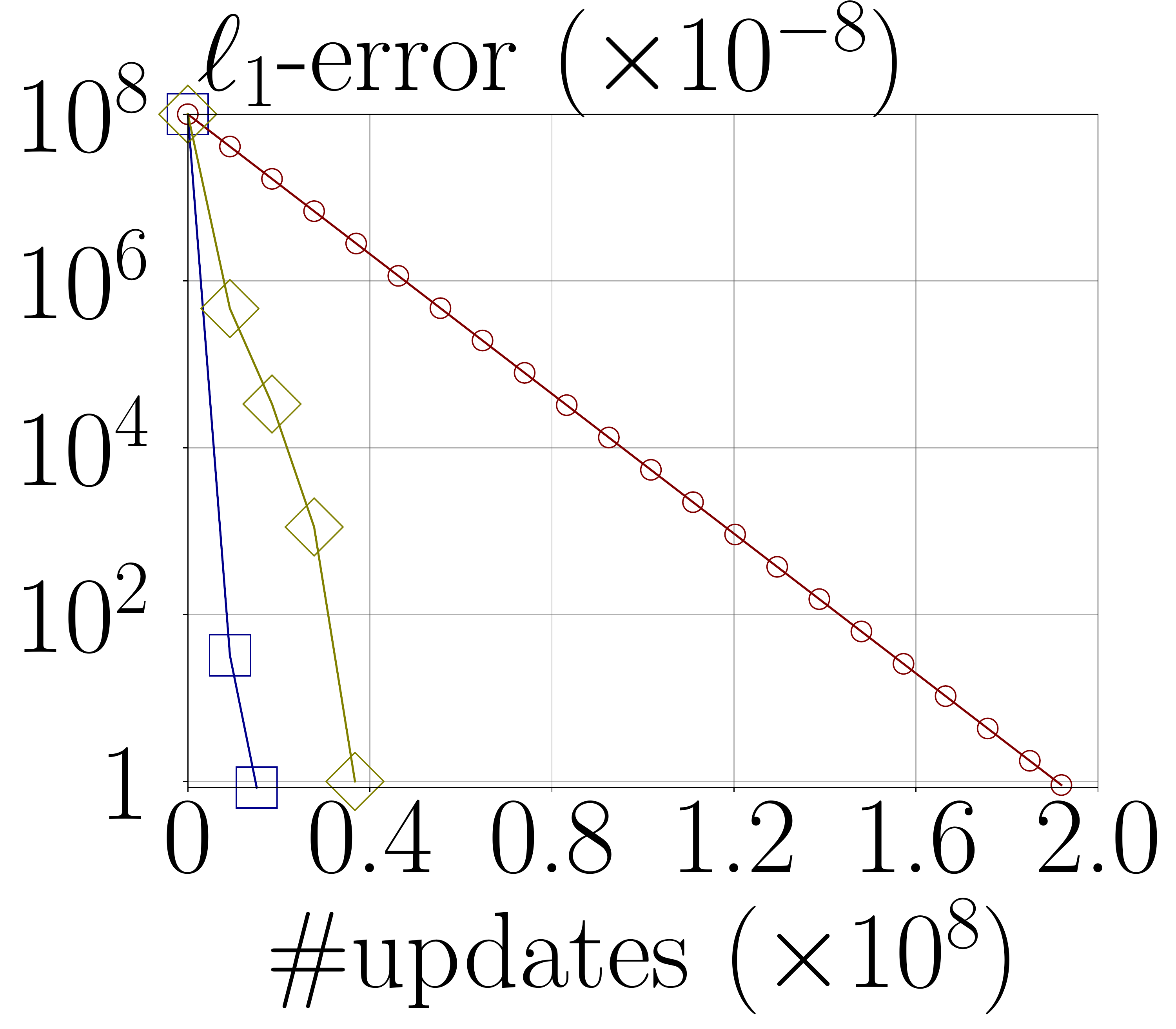} & 
   \hspace{-4.5mm}
	\includegraphics[width=0.18\linewidth]{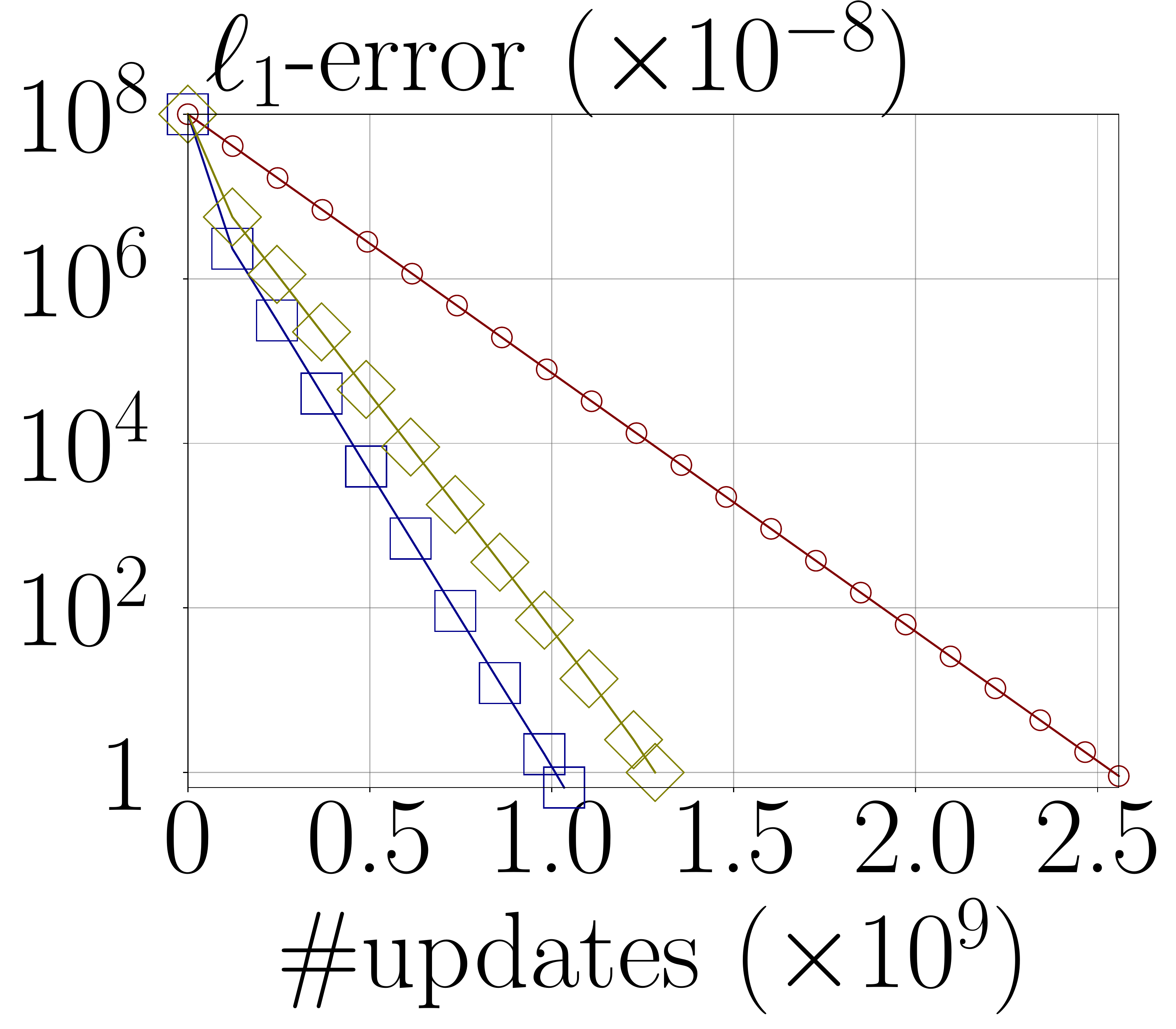} & 
   \hspace{-4.5mm}
	\includegraphics[width=0.18\linewidth]{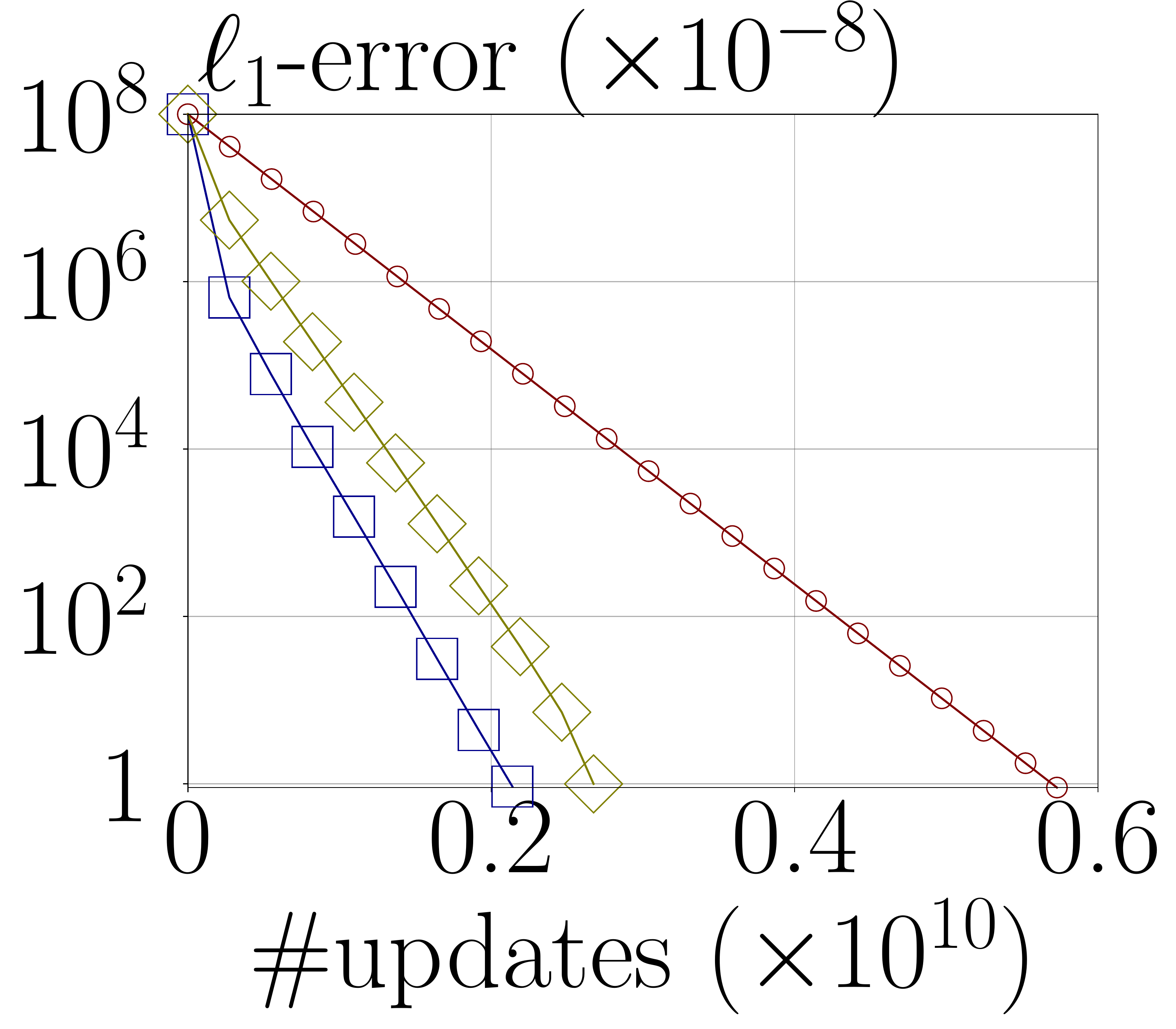} & 
   \hspace{-4.5mm}
	\includegraphics[width=0.18\linewidth]{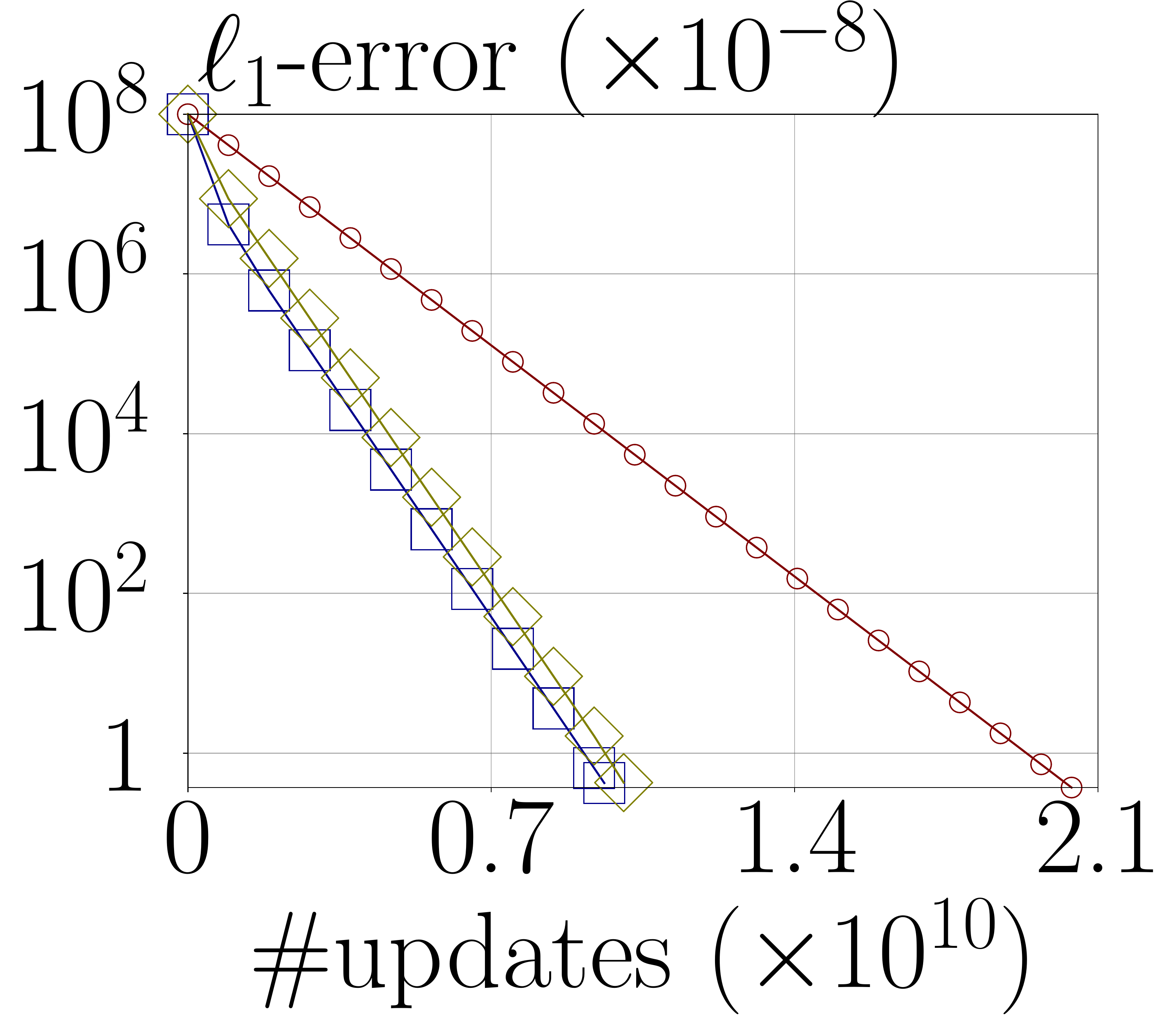} & 
   \hspace{-4.5mm}
   \includegraphics[width=0.18\linewidth]{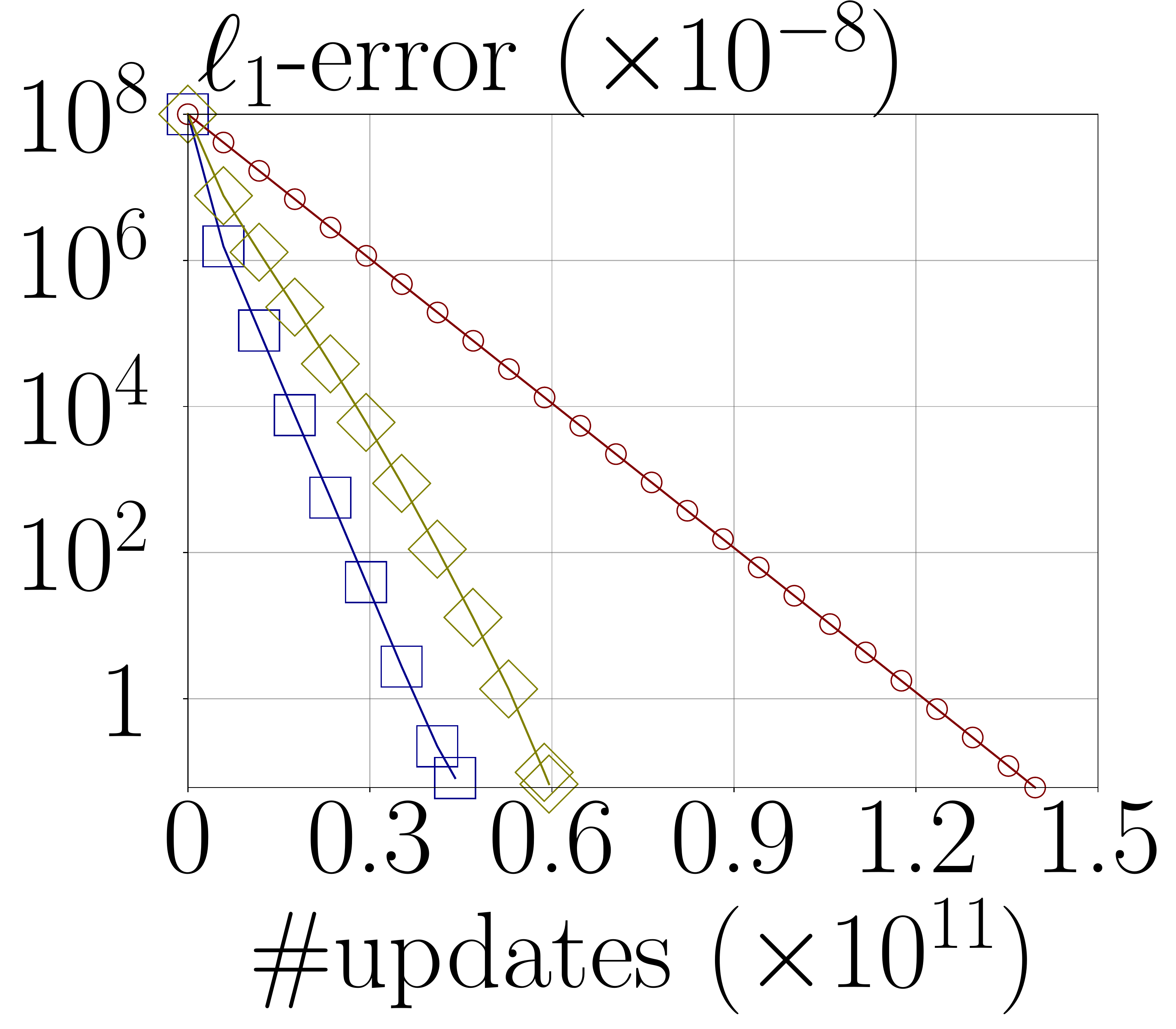} \\[-2mm]
	
	\hspace{-10mm} (a) {\em dblp} & \hspace{-4.6mm} (b) {\em web-Stanford} & \hspace{-4.6mm} (c) {\em pokec}  & \hspace{-4.6mm} (d)  {\em liveJournal}  & \hspace{-4.6mm}  (e)  {\em orkut}  & \hspace{-4.6mm} (f)  {\em twitter}  \\[1mm]
	
	\end{tabular}
}
	\vspace{-5mm}
     \caption{actual $\ell_1$-error  v.s. \#residue updates} 
     \label{fig:r_sum_vs_updates}
 	\vspace{-3mm}
 \end{figure*}

 \vspace{2mm}
\section{Experiments}\label{sec:exp}
 
In this section, we evaluate our proposed algorithms and verify our theoretical analysis with experiments. 

\vspace{1mm}
\noindent \textbf{Datasets}. We use six real datasets\footnote{~All these datasets could be found at \url{https://snap.stanford.edu/data/}}: 
\textit{DBLP}~\cite{YangL12}, 
\textit{Web Stanford (Web-St)}~\cite{LeskovecLDM09}, 
\textit{Pokec}~\cite{Takac12}, 
\textit{Live Journal (LJ)}~\cite{BackstromHKL06}, 
\textit{Orkut}~\cite{YangL12}, and 
\textit{Twitter}~\cite{KwakLPM10}.
These datasets have been commonly used in the experiments in the previous work \cite{WangTXYL16, WYXWY17, wei2018topppr, lin2020index, lofgren2015personalized}, on which the algorithm performance are considered as benchmarks.
While the graphs in {\em DBLP} and {\em Orkut} are un-directed, we replace each un-directed edge with two directed edges in both directions. 
For each dataset, we remove the isolated nodes, i.e., the nodes have no in-coming nor out-going edges; 
for the rest nodes, we relabel their id's with integers starting from $0$. 
Table~\ref{tab:datasets} shows the statistics of the datasets after the above cleaning process. 
Finally, in the experiments for evaluating the query efficiency, for each dataset, 
we perform queries on $30$ query source nodes generated uniformly at random for all the competitors and take the average query time.

\begin{table}[t]
\resizebox{\linewidth}{!}{%
\begin{tabular}{|l|c|c|c|c|c|c|}
\hline
\multicolumn{1}{|c|}{\multirow{3}{*}{Dataset}} & \multicolumn{3}{c|}{Index Size} & \multicolumn{3}{c|}{Construction Time} \\ \cline{2-7} 
\multicolumn{1}{|c|}{}                         & High-Prec.     & \multicolumn{2}{c|}{Approx.}  & High-Prec.       & \multicolumn{2}{c|}{Approx.}      \\ \cline{2-7}
\multicolumn{1}{|c|}{}                         & BePI     & FORA     & SpeedPPR  & BePI       & FORA       & SpeedPPR      \\ \hline
\textit{DBLP}                                  & 23.9MB   & 139MB    & 8.01MB    & 1.72       & 6.53       & 0.520         \\ \hline
\textit{Web-St}                                & 31.7MB   & 137MB    & 8.82MB    & 1.92       & 4.21       & 0.489         \\ \hline
\textit{Pokec}                                 & 1.13GB   & 1.24GB   & 118MB     & 75.4       & 248        & 16.2          \\ \hline
\textit{LJ}                                    & 2.32GB   & 3.31GB   & 263MB     & 185        & 612        & 38.8          \\ \hline
\textit{Orkut}                                 & 54.5GB   & 4.80GB   & 894MB     & 57988      & 1410       & 173           \\ \hline
\textit{Twitter}                               & 24.5GB   & 47.8GB   & 5.48GB    & 6180       & 19883      & 1256          \\ \hline
\end{tabular}
}
\caption{Index Size and Construction Time (in seconds)}
\label{tab:idx_size_and_time}
\hspace{-6mm}
\vspace{-8mm}
\end{table}

\vspace{1mm}
\noindent \textbf{Competitors.} 
There are two groups of competitors respectively for the experiments on high-precision and approximate SSPPR queries.
For the high-precision queries, we have the four competitors: $\powitr$, $\itrfwdpush$, $\powforpush$ and $\bepi$~\cite{jung2017bepi}: a state-of-the-art high-precision SSPPR algorithm which was reported that it outperforms most of (if not all) other existing works.
%
For the approximate queries, we compare the performance of the following competitors: $\speedppr$, $\speedppri$, $\fora$~\cite{YWXWLY019}, $\forai$~\cite{YWXWLY019}, and $\resacc$ \cite{lin2020index}: a most recent approximate SSPPR algorithm which was reported to have competitive performance comparing to $\fora$.

\vspace{1mm}
\noindent
{\bf Experiment Environment.}
All the experiments are conducted on a cloud based Linux 20.04 server with Intel 2.0 GHz CPU and 144GB memory. 
Except $\bepi$, all the competitors are implemented with C++, 
where 
the source code of the implementations of our algorithms can be found at here\footnote{~\url{https://github.com/wuhao-wu-jiang/Personalized-PageRank}} and 
the implementations of $\fora$, $\forai$ and $\resacc$ are open-source and provided by their respective authors.
Since only the MATLAB P-code\footnote{ A MATLAB file format that hides implementation details. } of $\bepi$ is released, 
we can only run $\bepi$ as a black box. 
All the C++ implementations are complied with GCC 9.3.0 with -O3 optimization.


\subsection{Evaluations of High-Precision SSPPR}

In this experiment, we evaluate the high-precision SSPPR algorithms. 
For $\powitr$, $\itrfwdpush$ and $\powforpush$, we set the $\ell_1$-error threshold $\lambda = \min \{10^{-8}, 1/ m\}$. 
$\bepi$ adopts a different error measurement, which is to compute the $\ell_2$ distance between the obtained results in two consecutive iterations, namely, $\|\hpis\itjnext - \hpis\itj\|_2= \sqrt{\sum_{v \in V} \left(\hpi\itjnext(s,v) - \hpi\itj(s,v))\right)^2}$;
when this $\ell_2$ distance is no more than a specified convergence parameter $\Delta$, it considers the current result $\hpis\itjnext$ converges and thus stops.
For $\bepi$, we set $\Delta = \min \{ 10^{-8}, 1/ m\}$.
It should be noted that under this setting of $\Delta$, 
the results obtained by $\bepi$ do not necessarily meet the requirement that 
the $\ell_1$-error (with respect to the ground truth $\pis$) is at most $\lmd$.
Therefore, its running time reported in the following experiments is an underestimate of $\bepi$'s actual time to achieve the $\ell_1$-error $\lmd$.

Moreover, among all these four competitors, $\bepi$ is the only one that requires pre-computed index.
Table~\ref{tab:idx_size_and_time} shows the pre-processing time and the index space consumption of $\bepi$~\footnote{~We save the pre-processing output in a .mat file and report file size as the index size.}.
$\bepi$ takes $57,988$ seconds (over $15$ hours) to compute the index on {\em Orkut} and $6,180$ seconds on {\em Twitter}, which consume $54.5$GB and $24.5$GB space, respectively. 
This is because $\bepi$ is a matrix-based algorithm and thus affected heavily by the density of the graph.
As shown in Table~\ref{tab:datasets}, the average degree of {\em Orkut} is $76.3$ while the one of {\em Twitter} is $35.3$.
Hence, the pre-processing time (rsp. index size)  of the former is significantly longer (rsp. larger) than that of the latter.

\vspace{1mm}
\noindent
{\bf Average Overall Query Time.}
Figure~\ref{fig:time_vs_datasets} reports the average overall running time of all the algorithms for the randomly generated query source nodes over all the datasets. 
The running time of $\powforpush$ is the smallest on all datasets except \textit{DBLP}, the dataset with fewest edges among the six, where 
$\powforpush$ is slightly worse than $\bepi$.
It is worth pointing out that even taking the advantages of a significant pre-processing (whose cost is not counted in the query time), $\bepi$ is still $2\times$ to $4\times$ slower than our $\powforpush$ in general. 
In particular, on {\em Orkut}, $\powforpush$ is $17\times$ faster than $\bepi$ even without any pre-processing or index.
This shows a significant superiority of $\powforpush$ over $\bepi$.
On the other hand, $\itrfwdpush$ and $\powitr$ have similar performance over all the datasets.
This is reasonable because they are essentially equivalent and having the same time complexity.
Interestingly, as  $\powforpush$ is carefully designed to incorporate both the strengths of $\powitr$ and $\itrfwdpush$, 
$\powforpush$ outperforms both of them in all cases.

\begin{figure*}
	 \vspace{-2mm}
	 \includegraphics[width=0.7\linewidth]{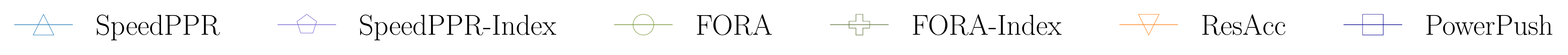} \\[-1mm]
\resizebox{.9\linewidth}{!}{
	 \begin{tabular}{cccccc} 
	\hspace{-6mm} 
	    \includegraphics[width=0.18\linewidth]{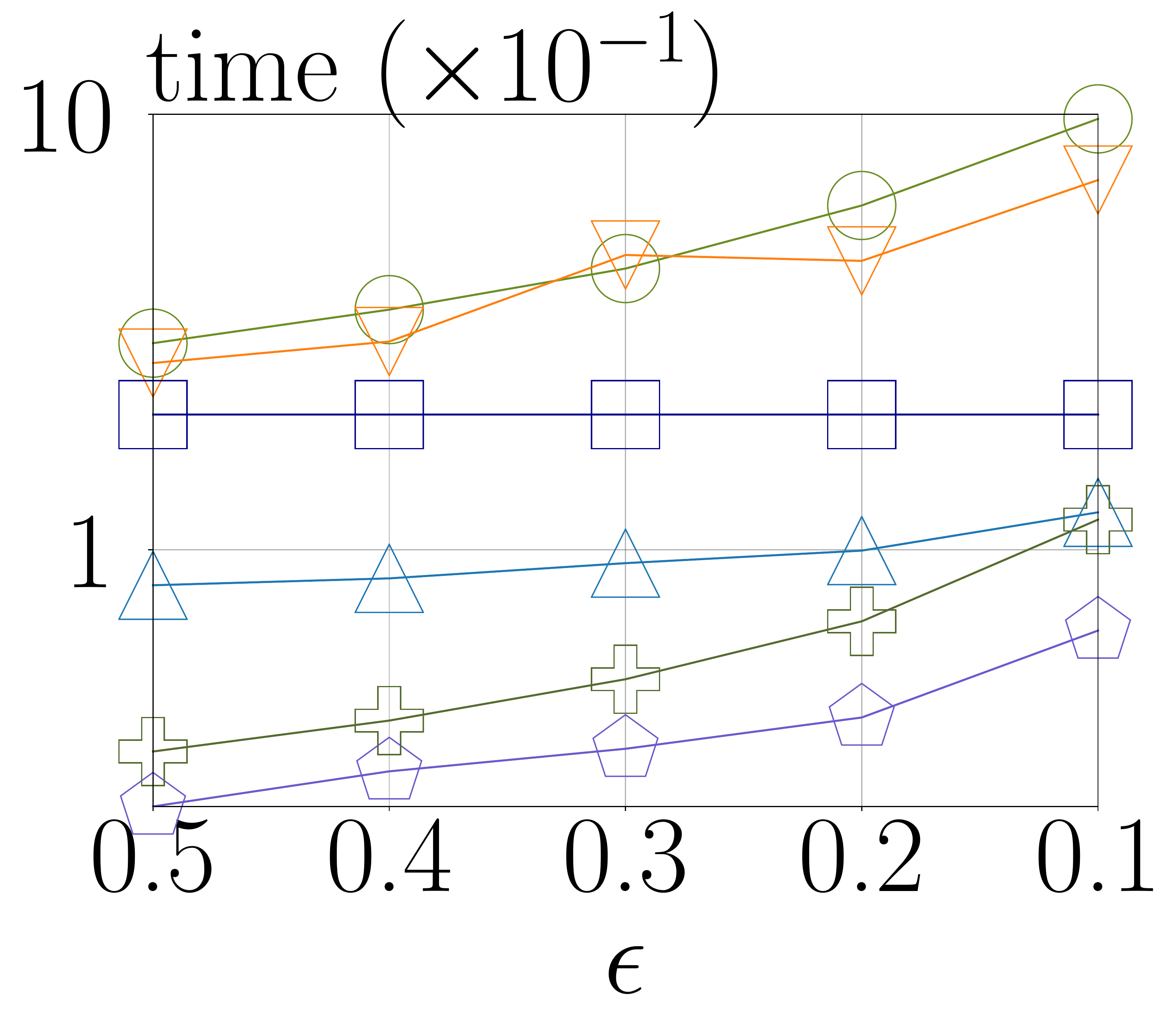} & 
   \hspace{-4mm}
\includegraphics[width=0.18\linewidth]{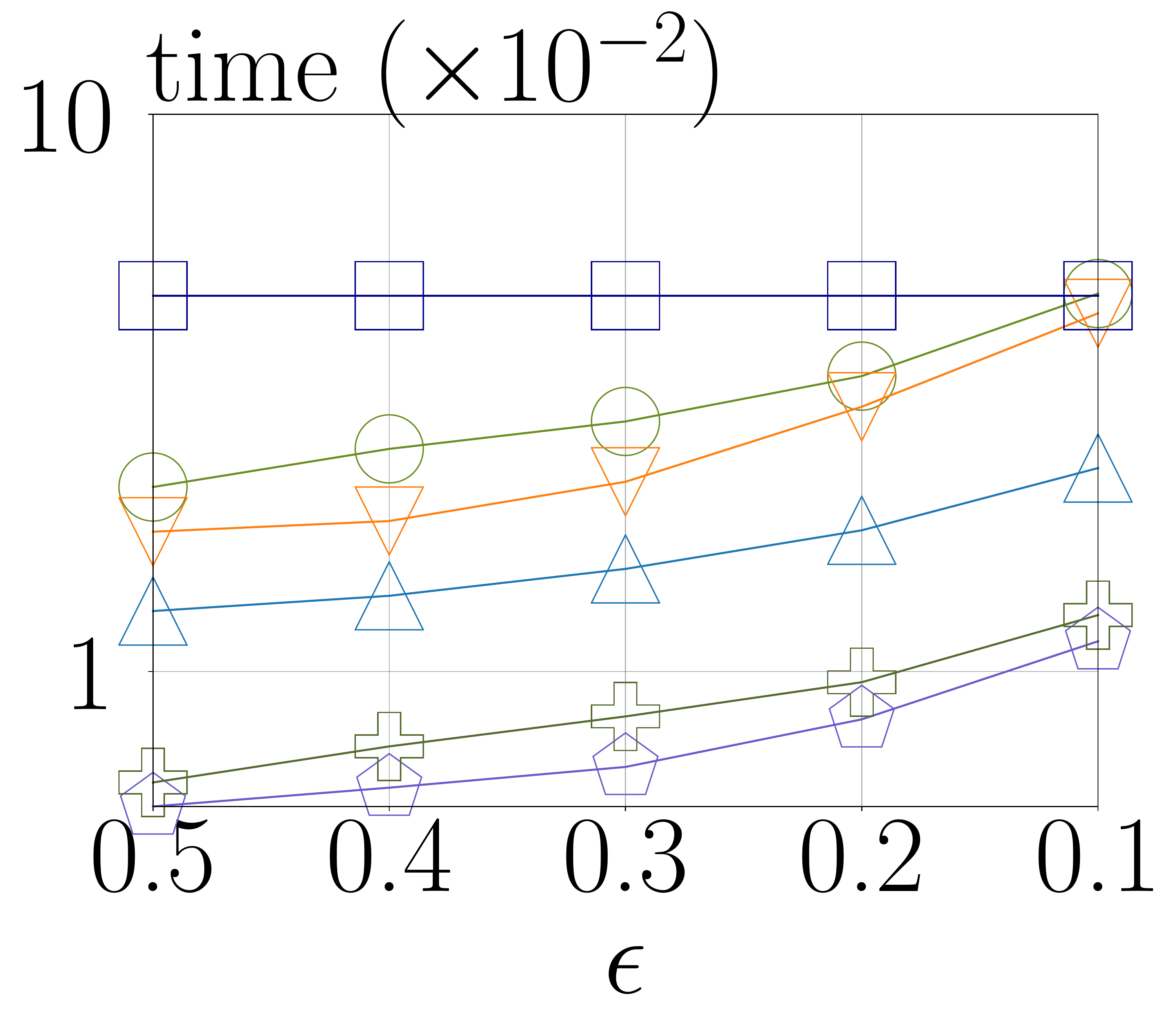} & 
   \hspace{-4mm}
	\includegraphics[width=0.18\linewidth]{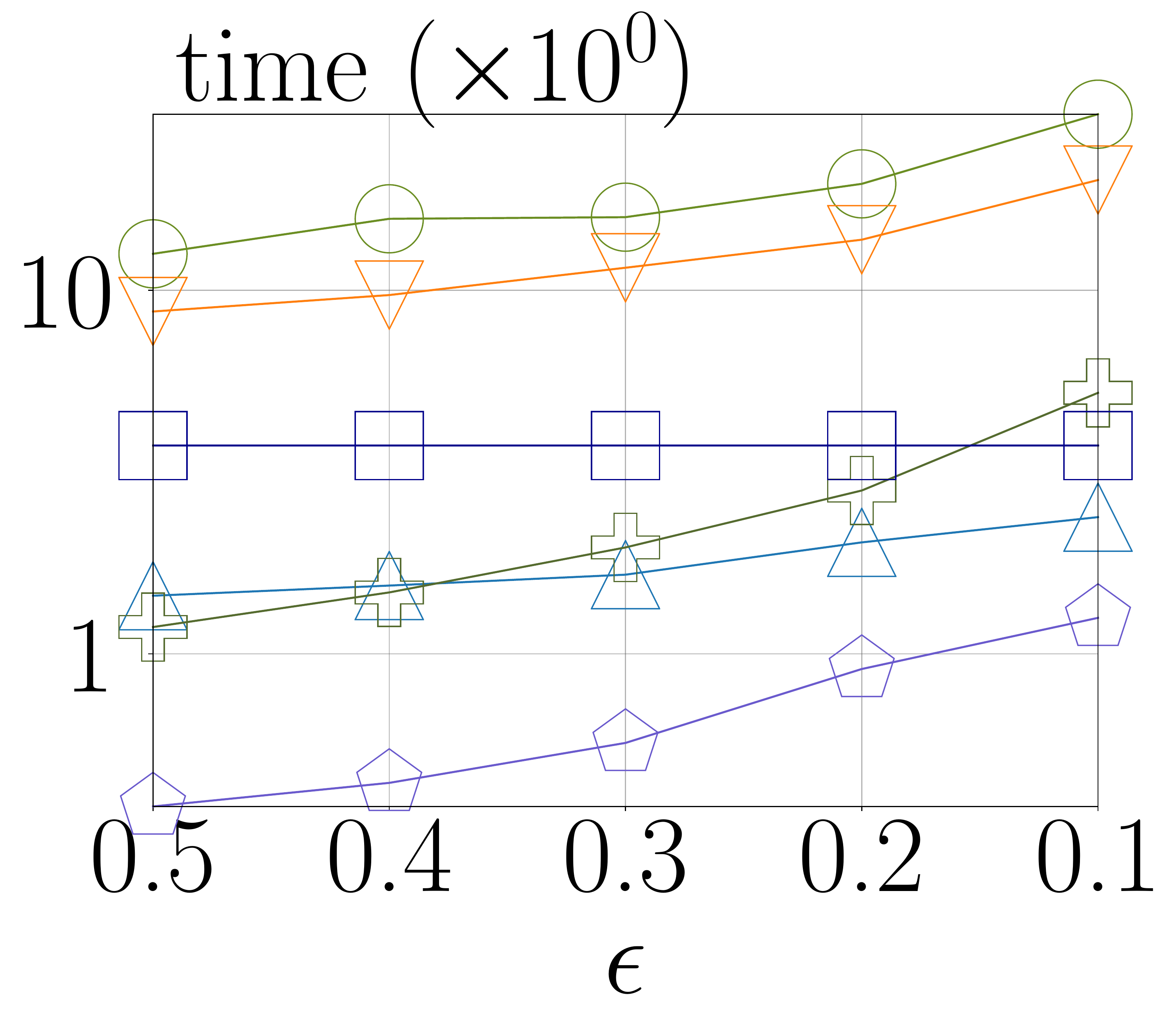} & 
   \hspace{-4mm}
	\includegraphics[width=0.18\linewidth]{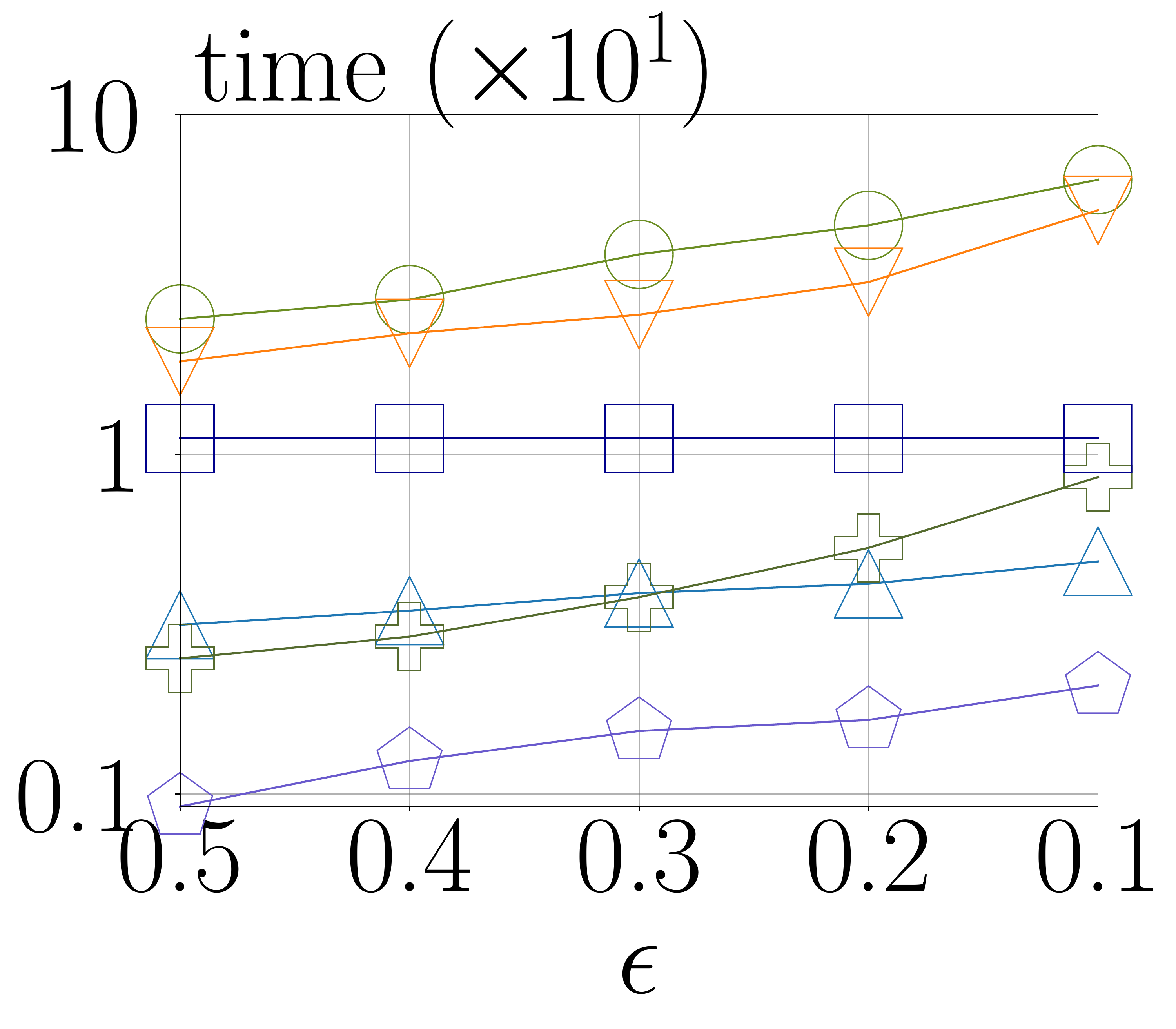} & 
   \hspace{-4mm}
	\includegraphics[width=0.18\linewidth]{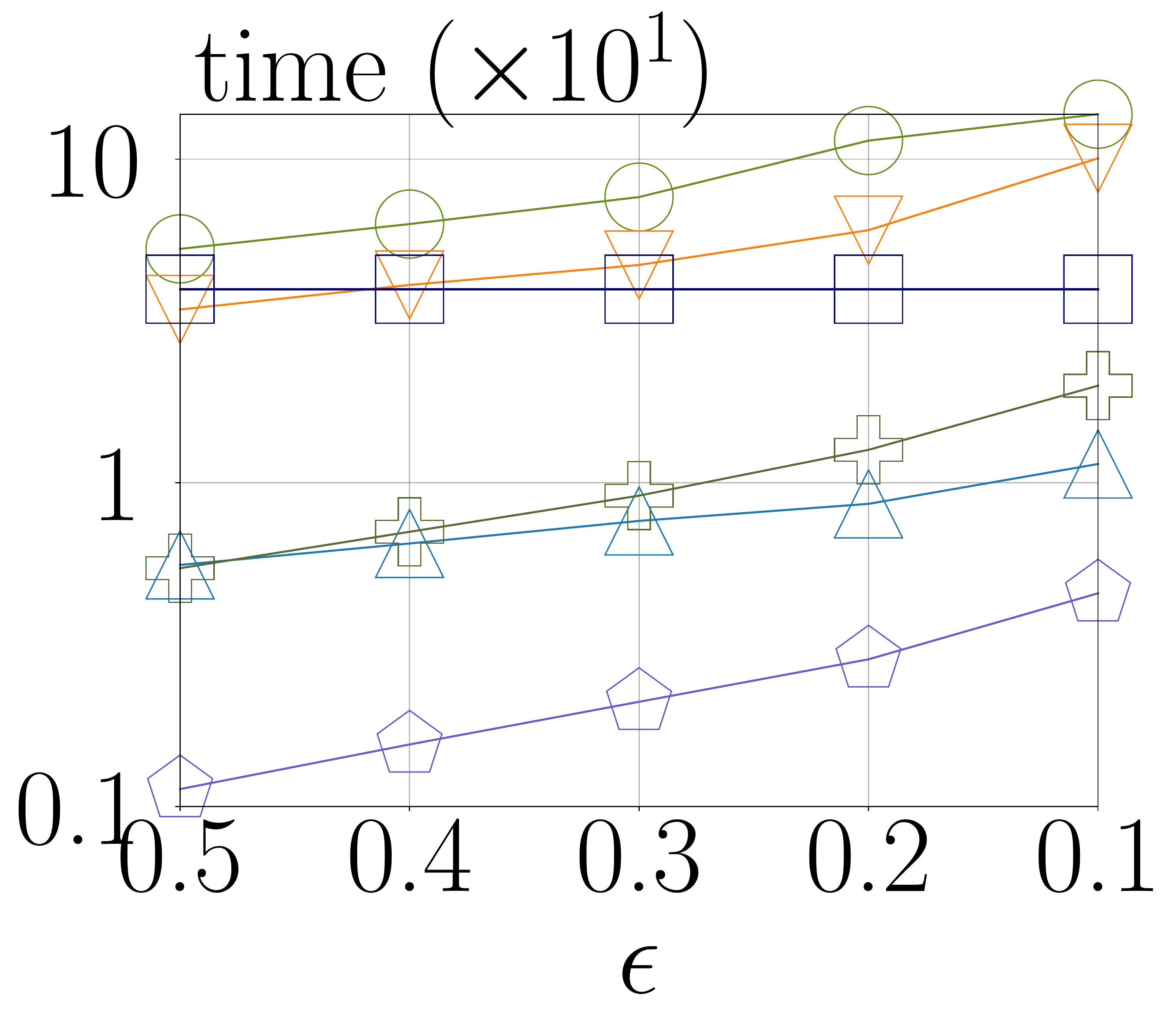} & 
   \hspace{-4mm}
   \includegraphics[width=0.18\linewidth]{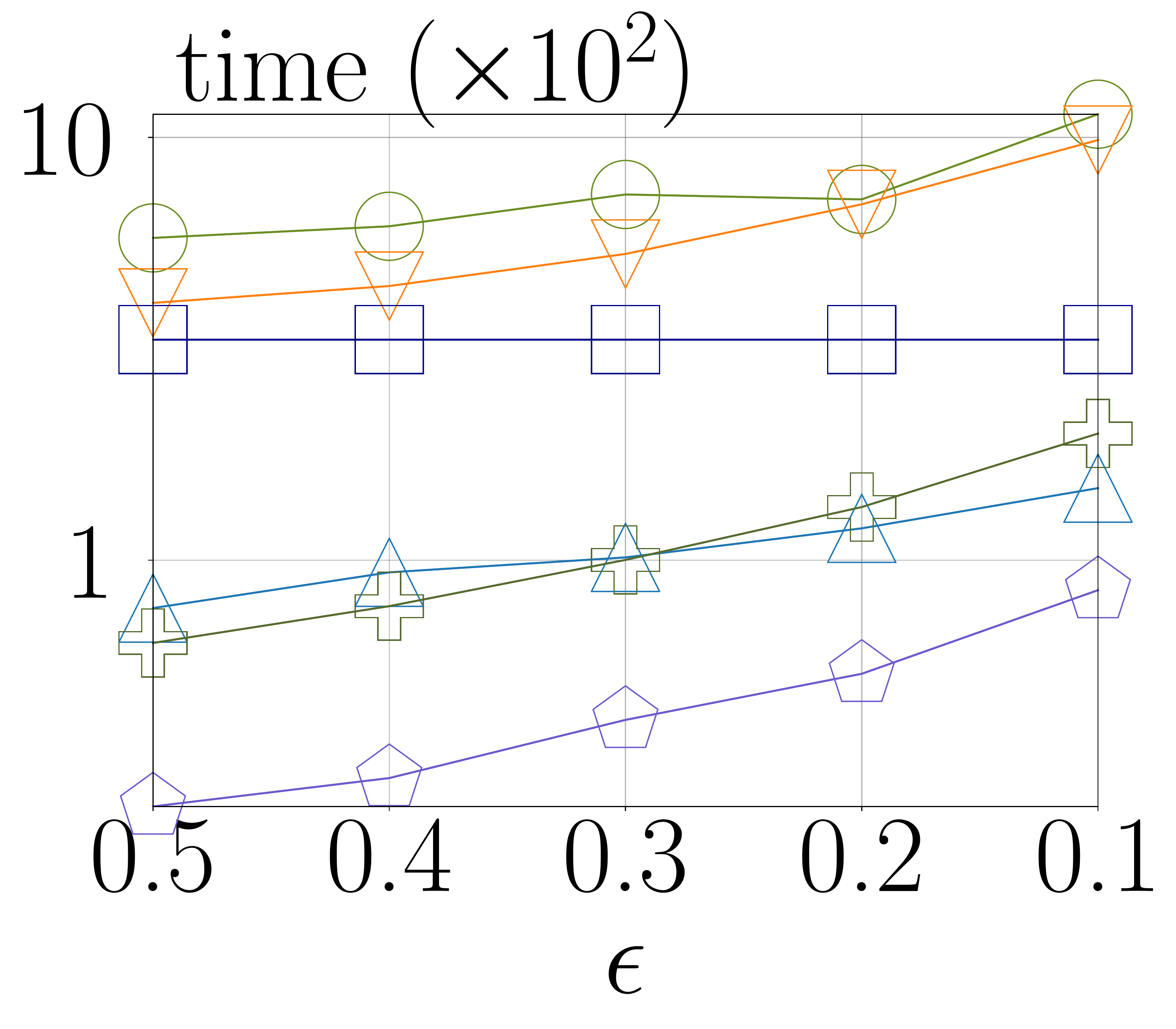} \\ [-2mm]
	\hspace{-10mm} (a) {\em dblp} & \hspace{-4.6mm} (b) {\em web-Standford} & \hspace{-4.6mm} (c) {\em pokec}  & \hspace{-4.6mm} (d)  {\em liveJournal}  & \hspace{-4.6mm}  (e)  {\em orkut}  & \hspace{-4.6mm} (f)  {\em twitter}  \\[1mm]
	\end{tabular}
}
	\vspace{-5mm} 
     \caption{running time (seconds) v.s. $\epsilon$} 
     \label{fig:time_vs_eps}
     \vspace{-2mm}
 \end{figure*}

 \begin{figure*}
	 \vspace{-2mm}
\resizebox{.9\linewidth}{!}{
	 \begin{tabular}{cccccc} 
	\hspace{-6mm} 
    \includegraphics[width=0.18\linewidth]{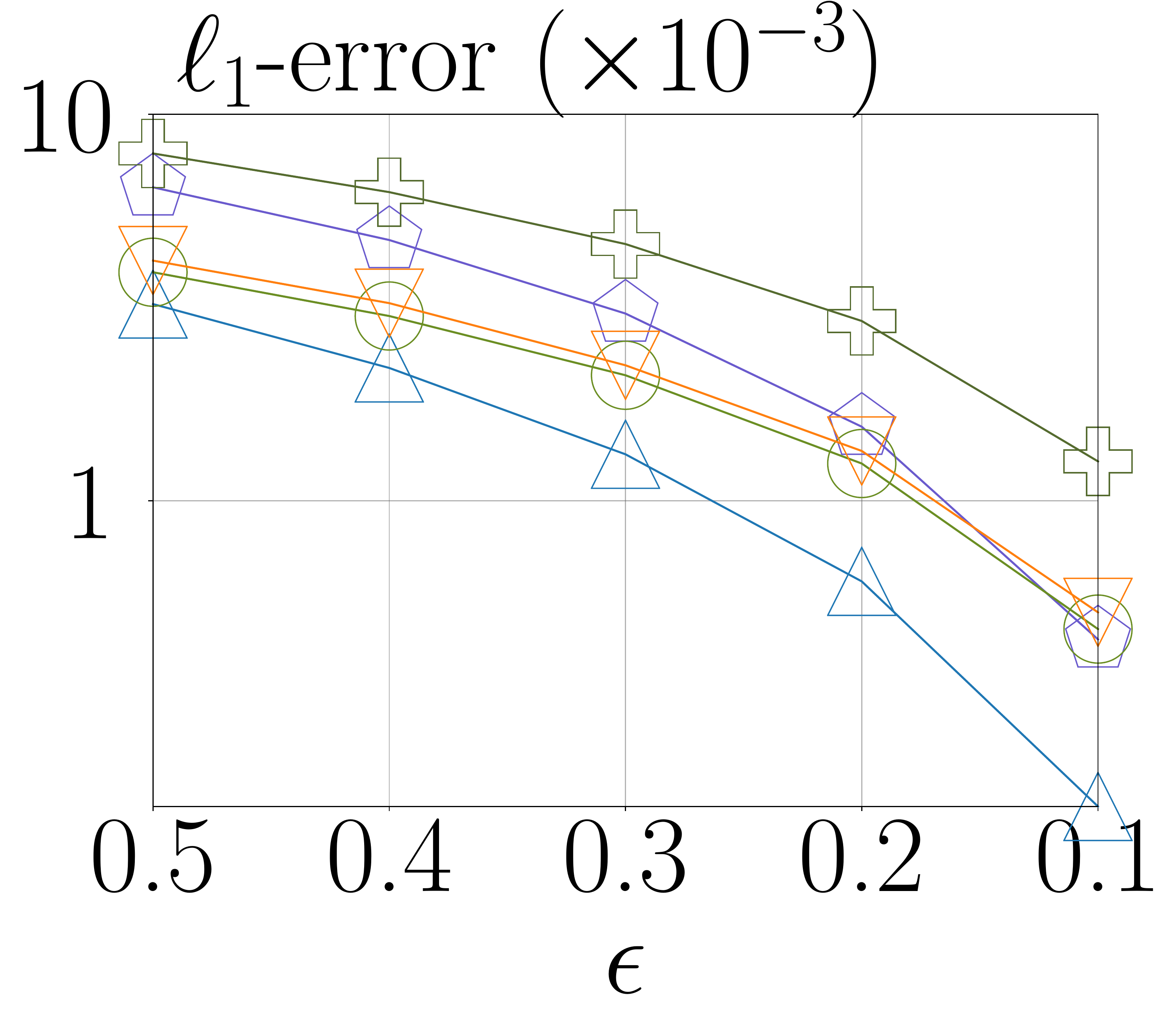} & 
   \hspace{-4mm}
	\includegraphics[width=0.18\linewidth]{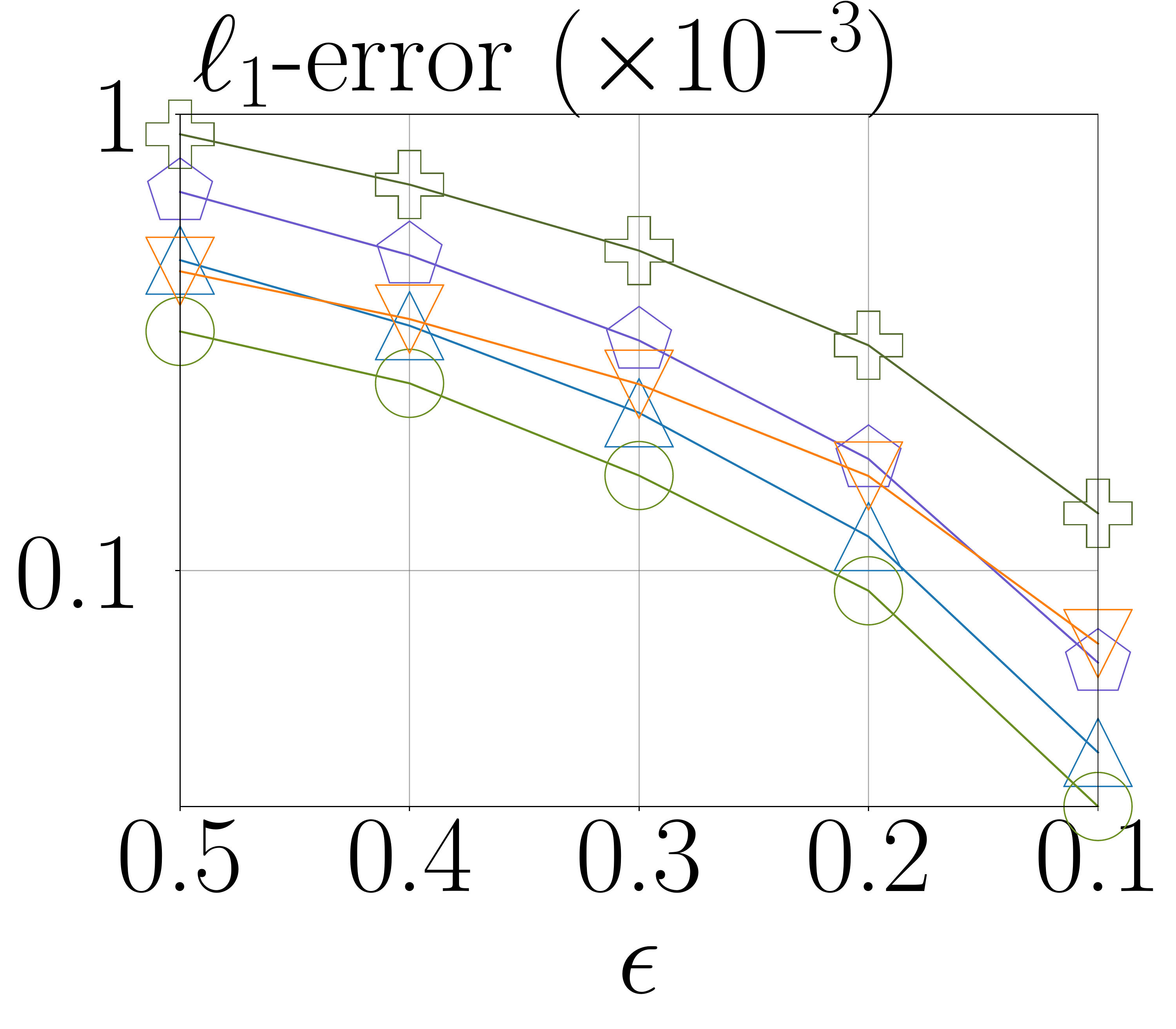} & 
   \hspace{-4mm}
	\includegraphics[width=0.18\linewidth]{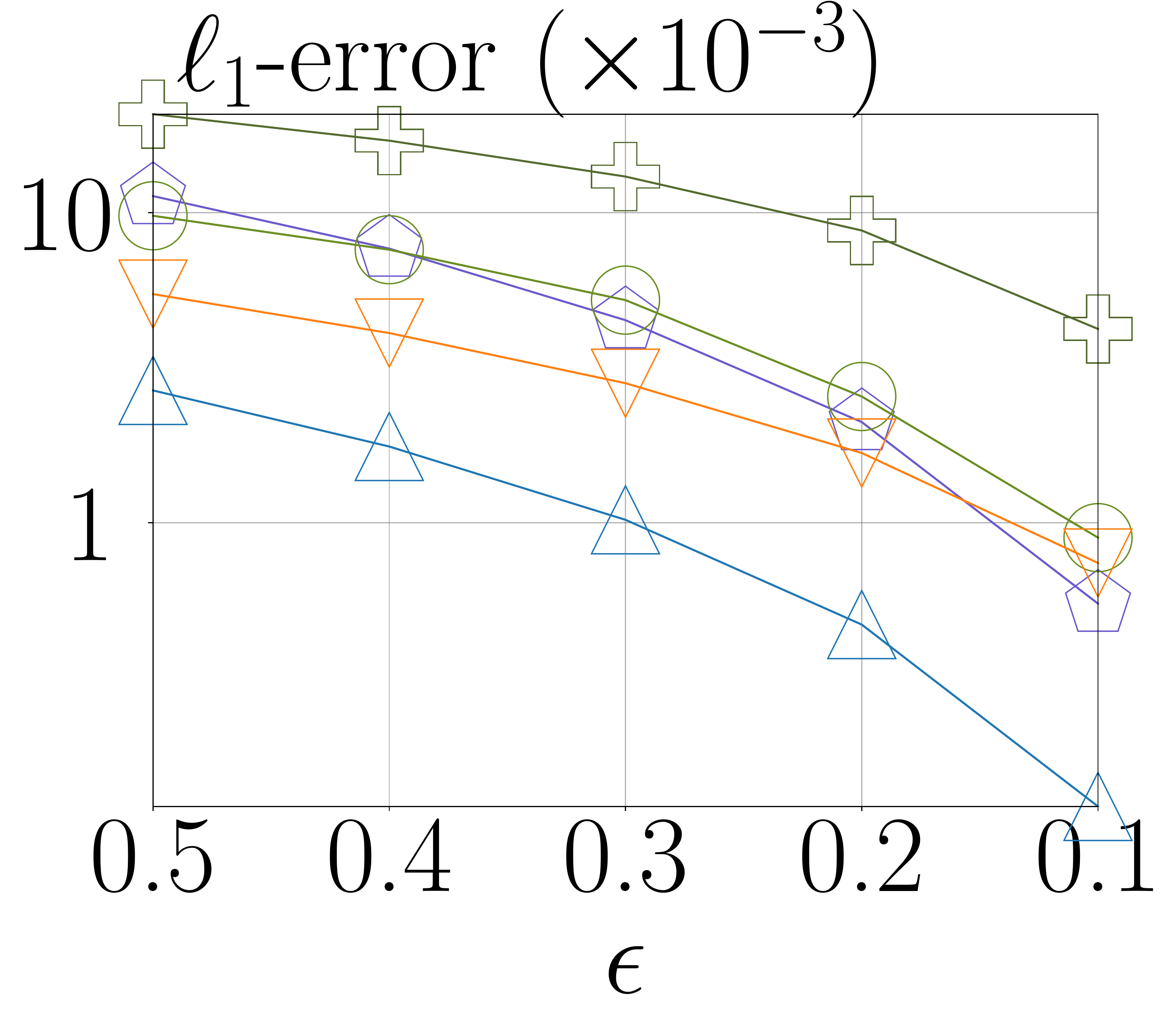} & 
   \hspace{-4mm}
	\includegraphics[width=0.18\linewidth]{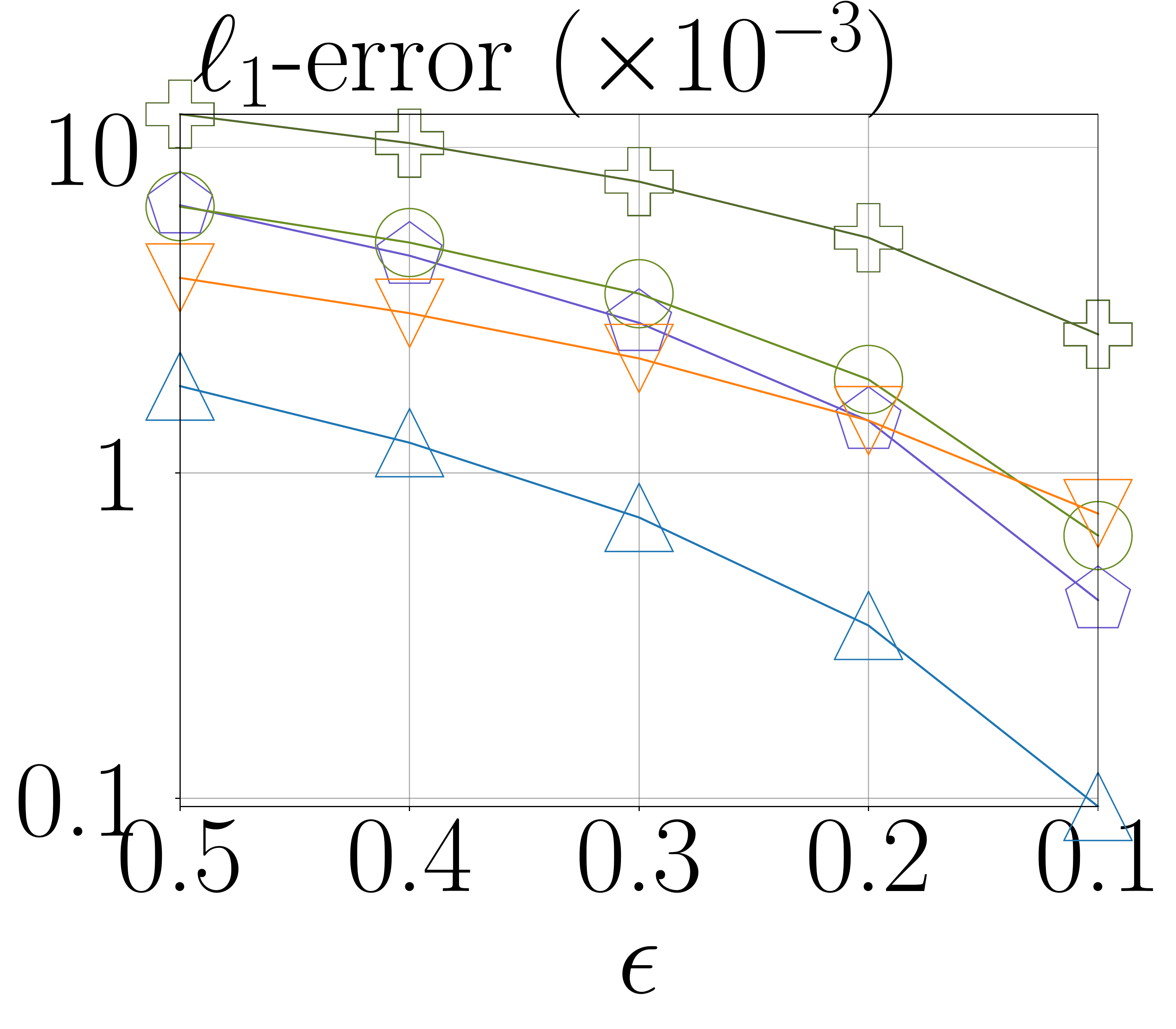} & 
   \hspace{-4mm}
	\includegraphics[width=0.18\linewidth]{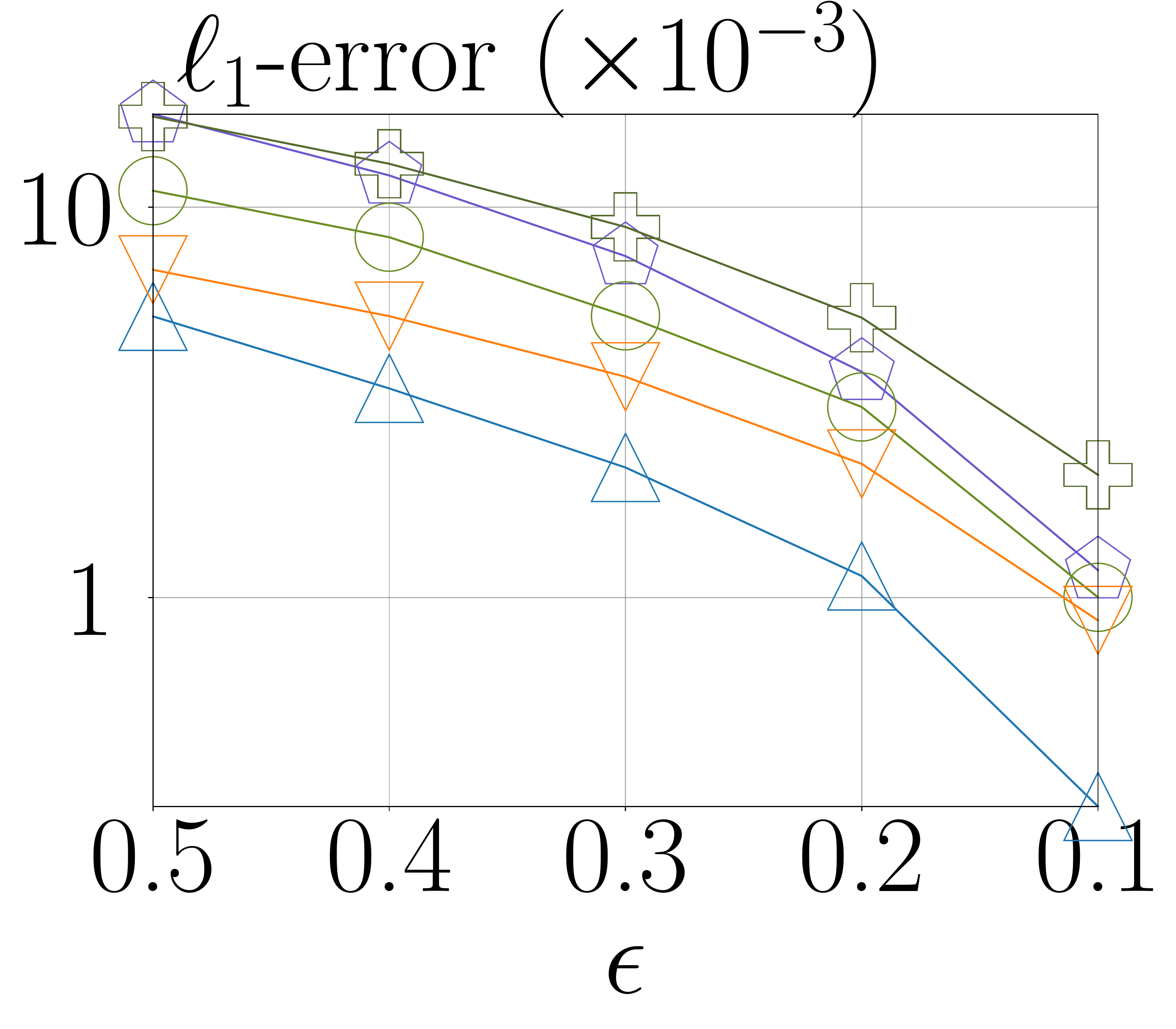} & 
   \hspace{-4mm}
   \includegraphics[width=0.18\linewidth]{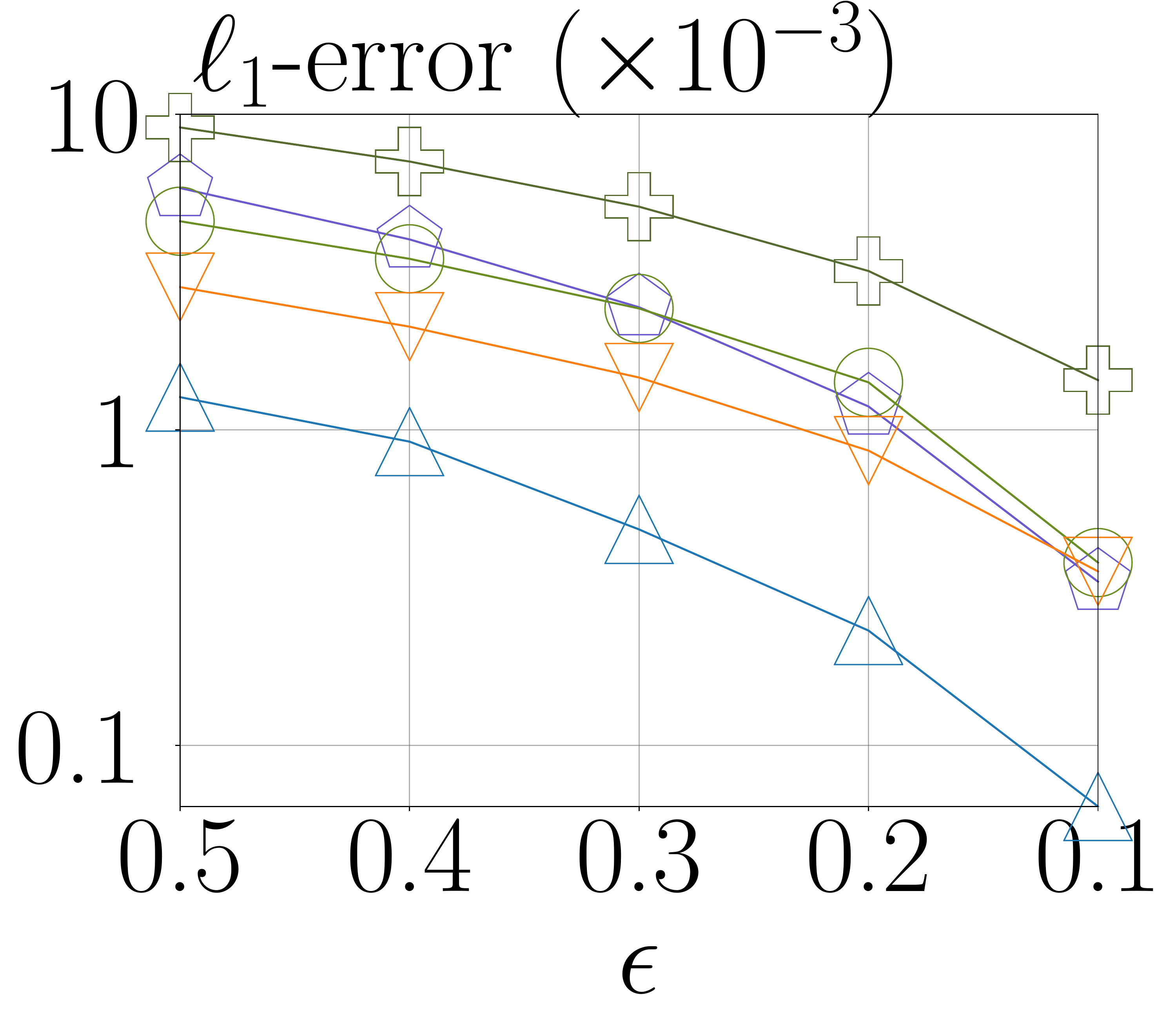} \\[-2mm]

	\hspace{-10mm} (a) {\em dblp} & \hspace{-4.6mm} (b) {\em web-Standford} & \hspace{-4.6mm} (c) {\em pokec}  & \hspace{-4.6mm} (d)  {\em liveJournal}  & \hspace{-4.6mm}  (e)  {\em orkut}  & \hspace{-4.6mm} (f)  {\em twitter}  \\[1mm]
	
	\end{tabular}
}
\vspace{-5mm} 
     \caption{actual $\ell_1$-error  v.s. $\epsilon$} 
     \label{fig:l1_vs_eps}
     \vspace{-3mm}
 \end{figure*}

\vspace{1mm}
\noindent
{\bf Actual $\ell_1$-Error v.s. Execution Time.}
Figure~\ref{fig:time_vs_r_sum} shows the {\em actual} $\ell_1$-error $\rsum$ (in log scale) versus the execution time of all the competitors. 
In this experiment, we take the query that incurs the median running time (among the 30 queries) of $\powforpush$ on each dataset as reference.
Each of these diagrams is plotted based on the execution with the corresponding  median query source node.
Except $\bepi$, 
the data points in the diagrams of each algorithm are plotted for the moments of every $4\cdot m$ edge pushing's (where each push operation on $v$ is counted as $d_v$ edge pushing's).
As $\bepi$ adopts different error measurement, we take a decreasing sequence of $\Delta$ values until $\Delta = \min\{1/m, 10^{-8}\}$  for $\bepi$ 
and compute the corresponding $\ell_1$-error for the obtained results, and plot these $\ell_1$-errors along with the corresponding execution time. 
In the diagrams, some curves of $\bepi$ do not touch the bottom; in these cases, 
$\bepi$ did not manage to obtain an estimation within $\ell_1$-error $\lmd$ under the corresponding parameter setting.

There are three crucial observations from Figure~\ref{fig:time_vs_r_sum}.
First, $\powforpush$ has the fastest convergence speed on all datasets, where it outperforms $\bepi$ by 
(i) an order of magnitude on {\em Orkut}, (ii) roughly two to four times on the other datasets except {\em DBLP}, and  
(iii) having roughly the same running time on {\em DBLP}. 
This is consistent with our observation from Figure~\ref{fig:time_vs_datasets}.
Second, except $\bepi$, the curves of the other three algorithms are pretty straight with the log-scale y-axis.
This implies that their $\ell_1$-errors decrease in an exponential speed with running time, and thus it matches
their $O(m \cdot \log \frac{1}{\lmd})$ time complexity.
Third, $\powitr$ has a faster convergence speed than $\itrfwdpush$ on four out of six datasets.
This is a bit counter-intuitive at the first glance.
But the reason for this is that after a few iterations, there would be a large number of active nodes.
In this case, the global sequential scan performs better than the random access in $\itrfwdpush$. 
This shows the importance of combining the global and local approach in $\powforpush$.

\vspace{1mm}
\noindent
{\bf Actual $\ell_1$-Error v.s. \# of Residue Updates.}
We further investigate the effectiveness of the push operations in the algorithms. 
Figure~\ref{fig:r_sum_vs_updates} 
demonstrates the $\ell_1$-error (in log scale) with respect to the number of edge pushing's, that is the number of residue updates.
Note that $\bepi$ is not applicable to this experiment, as we have no access to the operation number during its execution.
Except the first few updates, the log-scale $\ell_1$-errors of both $\itrfwdpush$ and $\powforpush$  decreases linearly.
This complies with our theoretical analysis. 
As expected, the pushes of $\itrfwdpush$ are more effective than those in $\powitr$, because they are performed in an  \textit{asynchronous} manner. 
Among the three algorithms, the proposed $\powforpush$ requires the least number of residue updates (to achieve the same $\ell_1$-error) in most datasets. 
This is because the {\em dynamic threshold} optimization enables $\powforpush$ to ``accumulate'' the residues of the nodes before pushing. 
And thus, it further reduces the number of the push operations.
Of interest is \textit{Orkut}, in which $\powforpush$ performs similar number of updates as $\itrfwdpush$. 
However, as shown in Figure~\ref{fig:time_vs_r_sum}, $\powforpush$ requires much less time than $\itrfwdpush$ on the same dataset. 
The reason is that the \textit{global sequential scan} technique makes the memory access pattern in $\powforpush$ more cache-friendly and hence more efficient to perform pushes. 
Similar observation can also be found in the comparison between $\powitr$ and $\itrfwdpush$ on {\em Orkut}, where
$\powitr$ performs a much larger number of operations but it achieves a similar execution time as $\itrfwdpush$'s.

\vspace{-2mm}
\subsection{Evaluations of Approximate SSPPR}

	
	
 
Next, 
we evaluate the approximate SSPPR algorithms 
against different $\eps$ values from $0.1$ to $0.5$, 
and report their running time as well as the solution quality in terms of $\ell_1$-error. 
For the index version of $\fora$, we generate its index with the smallest $\epsilon$ in consideration, i.e., $\epsilon = 0.1$, 
and re-use it for other $\epsilon$'s. 
For the index-based $\speedppr$, its index size does not depend on $\epsilon$. 
As shown in Table~\ref{tab:idx_size_and_time}, $\speedppr$ outperforms $\fora$ in both pre-processing time and index size by an order of magnitude.

\vspace{1mm}
\noindent
{\bf Running Time v.s. $\eps$.}
Figure \ref{fig:time_vs_eps} shows the running time (in log scale) of 
all the competitors over the six datasets. 
Note that we deliberately include our high-precision algorithm $\powforpush$ in these diagrams as a base line. 
Interestingly, it shows comparable or even better performance comparing to the state-of-the-art index-free approximate algorithms ($\fora$ and $\resacc$) on some datasets. 
Furthermore, observe that $\speedppri$ demonstrates superior performance over all datasets. 
The index-free version of $\speedppr$ is slightly slower than $\forai$. 
Indeed, 
except for the two smallest datasets, 
the efficiency of $\speedppr$ is comparable or even better than that of $\forai$ with small $\epsilon$'s. 
Both $\speedppr$ and $\speedppri$ show a linear increase on the running time (in log scale), 
especially on \textit{Orkut} and \textit{Twitter}. 


\vspace{1mm}
\noindent
{\bf Actual $\ell_1$-Error v.s. $\eps$.}
Finally, we study the solution quality of the approximate algorithms. 
Figure~\ref{fig:l1_vs_eps} shows the $\ell_1$-error with respect to the ground truth $\pis$ which is computed with $\powforpush$ by setting $\lmd = 10^{-17}$, the highest possible precision for the data type {\em double} in C++.  
Except on the dataset \textit{web-Stanford}, $\speedppr$ offers the best solution quality. 
When $\epsilon$ is small, its solution quality could be an order of magnitude better than other algorithms. 
This is impressive, considering it just takes comparable running time of $\forai$. 
Another observation is that both $\speedppri$ and $\forai$ provide inferior solutions compared to the index-free algorithms. 
The reason for this is that the index-based algorithms tend to use more random walks as the walks can be performed with a relatively small cost.
These algorithms  thus spend less time on the local push phase, which actually computes the estimation {\em deterministically}.
As a result, the random walks are performed based on a larger $\rsum$ leading to a larger variance in the estimations.


%% file: conclusion.tex
\section{Conclusion}\label{sec:conclusion}

%
In this paper, we show an equivalent connection between the two fundamental algorithms $\powitr$ and $\fwdpush$. 
Embarking from this connection, we further prove that the time complexity of a common $\fwdpush$ implementation is $O(m \cdot \log \frac{1}{\lmd})$, where $\lmd$ is the $\ell_1$-error threshold. 
This answers the long-standing open question regarding the time complexity of $\fwdpush$ in the dependency on $\lmd$.
Based on this finding, we propose a new implementation of $\powitr$, called $\powforpush$, which incorporates both the strengths of $\powitr$ and $\fwdpush$.  
Furthermore, we propose a new algorithm, called $\speedppr$ for answering approximate single-source PPR queries. 
The expected time complexity of $\speedppr$ is $O(n\log n \log \frac{1}{\eps})$ on scale-free graphs, improving the state-of-the-art $O(\frac{n\log n}{\eps})$-bound.
In addition, $\speedppr$ admits an index with size always at most $O(m)$ independent on $\eps$.
Our experimental results show that our $\powforpush$ and $\speedppr$ outperform their state-of-the-art competitors by up to an order of magnitude in all evaluation metrics.